\documentclass[12pt, draftcls, onecolumn]{IEEEtran}

\usepackage{amsmath, mathrsfs,amsfonts}
\usepackage{amssymb}
\usepackage{acronym}
\usepackage{graphicx, amsthm}
\usepackage{color}

\newcommand{\BEQA}{\begin{eqnarray}}
\newcommand{\EEQA}{\end{eqnarray}}

\newcommand{\bqed}{\nopagebreak\mbox{}\hfill$\blacksquare$\smallskip}
\newenvironment{bproof}{\noindent{\bf Proof}\hspace*{0.1em}}{\bqed}

\def\e{\epsilon}
\def\a{\alpha}
\def\t{\theta}
\def\p{\hat p}
\def\ph{\hat p}

\def \foral {\textrm{for all }}
\def \pr  {\mathbf{P}}
\def \R  {\mathbb{R}}
\def \E  {\mathbf{E}}

\def \foral {\textrm{for all }}
\def \pr  {\mathbf{Pr}}
\def \E  {\mathbf{E}}

\newtheorem{theorem}{Theorem}

\newtheorem{corollary}{Corollary}
\newtheorem{lemma}{Lemma}

\theoremstyle{definition}
\newtheorem{assumption}{Assumption}
\newtheorem{definition}{Definition}

               {\begin{list}{}{\leftmargin#1\rightmargin#2\topsep#3}\item{}}%
               {\end{list}}

\title{QoE-aware Media Streaming in Technology and Cost Heterogeneous Networks}

\author{Ali ParandehGheibi,  Asuman Ozdaglar and Muriel M\'edard\thanks{The authors are with the Department of Electrical Engineering and Computer Science, Massachusetts Institute of Technology, Cambridge, MA 02139 USA. (emails: \{parandeh, asuman, medard\}@mit.edu). {This paper was presented in part at the 49th IEEE Conference on Decision and Control.}}}

\begin{document}

\maketitle

\begin{abstract}

We present a framework for studying the problem of media streaming in technology and cost heterogeneous environments. We first address the problem of efficient streaming in a technology-heterogeneous setting. We employ random linear network coding to simplify the packet selection strategies and alleviate issues such as duplicate packet reception. Then, we study the problem of media streaming from multiple cost-heterogeneous access networks. Our objective is to characterize analytically the trade-off between access cost and user experience. We model the Quality of user Experience (QoE) as the probability of interruption in playback as well as the initial waiting time. We design and characterize various control policies, and formulate the optimal control problem using a  Markov Decision Process (MDP) with a probabilistic constraint. We present a characterization of the optimal policy using the Hamilton-Jacobi-Bellman (HJB) equation. For a fluid approximation model, we provide an exact and explicit characterization of a threshold policy and prove its optimality using the HJB equation.

Our simulation results show that under properly designed control policy, the existence of alternative access technology as a complement for a primary access network can significantly improve the user experience without any bandwidth over-provisioning.

\end{abstract}


\section{Introduction}\label{introduction_sec}

Media streaming is fast becoming the dominant application on the Internet \cite{Lab09}. The popularity of available online content has been accompanied by the growing usage of wireless handheld devices as the preferred means of media access.  The predictions by Cisco Visual Networking Index \cite{cisco_VNI} show that by 2015,  the various forms of video (TV, VoD, Internet Video, and P2P) will exceed 90 percent of global consumer traffic, and 66 percent of total mobile traffic. In order to cope with the demand, the wireless service providers generally build new infrastructure. Another approach that seems to be gaining momentum in the industry is offloading mobile data traffic onto another network through dual-mode such as additional WiFi interfaces \cite{cisco_mobile}. This approach requires the wireless devices to operate seamlessly in an environment with heterogeneous access methods (WiFi, 3G and 4G) and different access costs (cf. Figure \ref{fig:MultiServer}). For example, accessing public WiFi networks is free but unreliable, while there are additional charges associated with more reliable 3G or 4G data networks. The goal of this work is to design network access policies that minimize the access cost in such heterogeneous environments, while guaranteeing an acceptable level of quality of user experience.

\begin{figure}[ht]
\centering
\includegraphics[scale=.7]{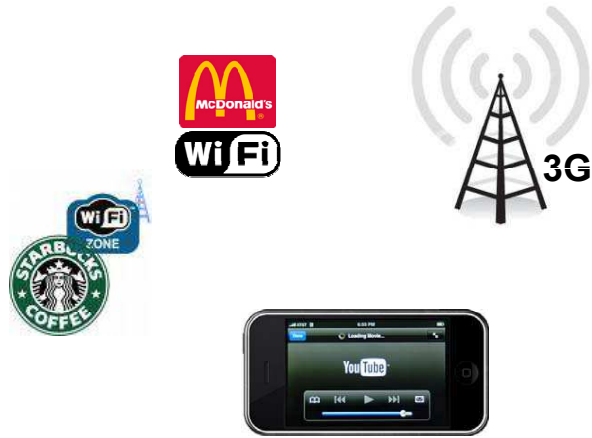}
  \caption{Media streaming from heterogeneous servers.}\label{fig:MultiServer}
\end{figure}


In particular, we focus on media streaming applications which are intrinsically delay-sensitive. Hence, they need to be managed differently from the traditional less delay-sensitive applications such as Web, Email and file downloads. Most of the current approaches for providing a reasonable Quality of Service (QoS) for streaming applications are based on resource over-provisioning, without considering the transient behavior of the service provided to such applications.
In this work, we pay special attention to communication and control techniques  that are specifically designed with Quality of user Experience (QoE) in mind. The goal, on one hand, is to make the optimal use of the limited and possibly unreliable resources to provide a seamless experience for the user. On the other hand, we would like to provide a tool for the service providers to improve their service delivery (specifically for delay-sensitive applications) in the most economical way. Our contributions are summarized  in the following.

We first address the problem of efficient streaming in technology-heterogeneous settings, where a user receives a stream over multiple paths from different servers. Each sever can be a wireless access point or another peer operating as a server. We consider a model that the communication link between the receiver and each server is unreliable, and hence, it takes a random period of time for each packet to arrive at the receiver from the time that the packet is requested from a particular server. One of the major difficulties with such multi-server systems is the packet tracking and duplicate packet reception problem, i.e., the receiver need to keep track of the index of the packets it is requesting from each server to avoid requesting duplicate packets. Since the requested information is delay sensitive, if a requested packet does not arrive within some time interval, the receiver need to request the packet from another server. This may eventually result in receiving duplicate packets and waste of resources. We address this issue and discuss that using random linear network coding  (RLNC)\cite{RLNC} across packets within each block of the media file we can alleviate this issue. This technique assures us that, with high probability, no redundant information will be delivered to the receiver.

We would like to emphasize that one of the critical roles of network coding techniques in this work, other than efficient and seamless streaming, is to  simplify greatly the communication models, so that we can focus on end-user metrics and trade-offs. For example, if each server can effectively transmit packets according to an independent Poisson process, using RLNC we can merge these processes into one Poisson process of sum rate. Hence, the system model boils down to a single-server single-receiver system.

We then study the problem of media streaming in a cost-heterogeneous environment. We consider a system wherein network coding is used to ensure that packet identities can be ignored, and packets may potentially be obtained from two classes of servers with different rates of transmission.  The wireless channel is unreliable, and we assume that each server can deliver packets according to a Poisson process with a known rate. Further, the costs of accessing the two servers are different; for simplicity we assume that one of the servers is free.  Thus, \emph{our goal is to develop an algorithm that switches between the free and the costly servers in order to satisfy the desired user experience at the lowest cost.}

The user experience metrics that we consider in this work are the initial buffering \emph{delay} before the media playback, and the probability of experiencing an interruption throughout media playback. Interruption probability captures the reliability of media playback. Such metrics best capture the user experience for most of media streaming applications e.g. Internet video, TV, Video on Demand (VoD), where the user may tolerate some initial delay, but expects a  smooth sequential playback. In \cite{JSAC}, we characterized the optimal trade-off between these metrics for a single-server single-receiver system. In particular, we established the following relation
\begin{equation}\label{int_exp}
    \textrm{Probability  of  interruption} = e^{-I(R) \cdot \textrm{(initial buffering)}},
\end{equation}
where $I(R)$ is the \emph{interruption exponent} or reliability function, which depends the arrival rate $R$ of the stream. This result is analogous to information theoretic error exponent results relating the error probability of a code to the block length of that code.

In a cost-heterogeneous system, the user experience such as initial waiting time may be improved by simultaneously accessing free and costly access methods. This adds another dimension to the problem the end-user is facing. Certain levels of user satisfaction can only be achieved by paying a premium for extra resource availability. Figure \ref{fig:3dim} illustrates a conceptual three-dimensional cost-delay-reliability trade-off curve.

\begin{figure}[ht]
\centering
\includegraphics[scale=0.9]{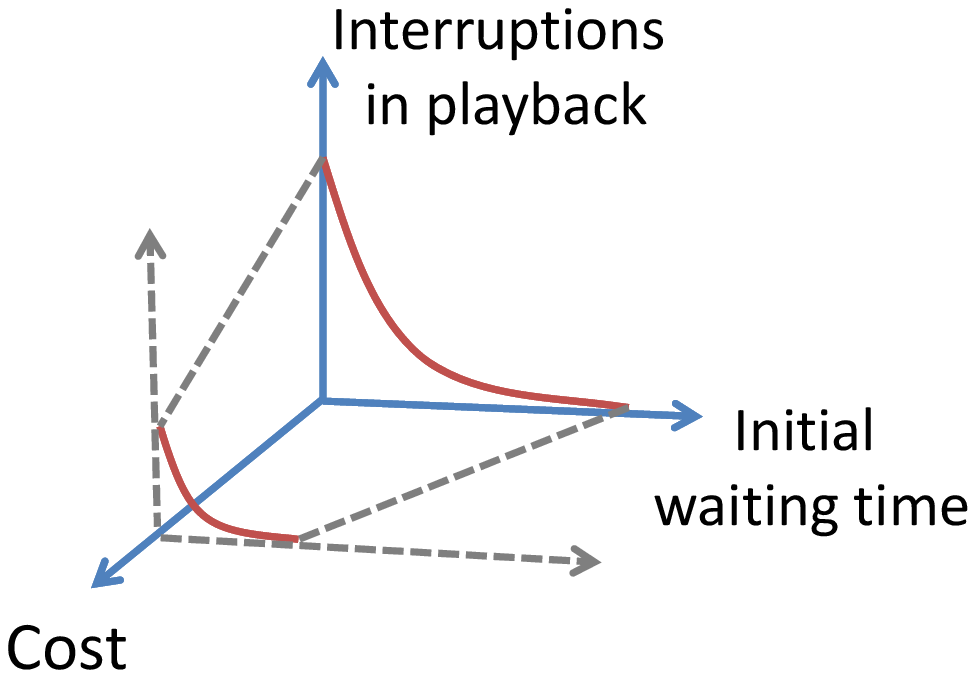}
  \caption{Trade-off between the achievable QoE metrics and cost of communication.}\label{fig:3dim}
\end{figure}

The objective of this paper is to understand the trade-off between initial waiting time, and the usage cost for attaining a target probability of interruption, and design control policies to achieve the optimal trade-off curve.  We study several classes of server selection policies.  Using the QoE trade-offs for a single-server system, we obtain a lower bound on the cost of offline policies that do not observe the trajectory of packets received.  We show that such policies have a threshold form in terms of the time of association with the costly server.  Using the offline algorithm as a starting point, we develop an \emph{online} algorithm with lower cost that has a threshold form -- both free and costly servers are used until the queue length reaches a threshold, followed by only free server usage.  We then develop an online risky algorithm in which the risk of interruption is spread out across the trajectory.  Here, only the free server is used whenever the queue length is above a certain threshold, while both servers are used when the queue length is below the threshold. The threshold is designed as a function of the initial buffering and the desired interruption probability. We numerically compare the performance of the proposed control policies. Our simulation results show online risky algorithm performs the best. Moreover, we observe that the existence of costly networks as a complement for unreliable but cheaper networks can significantly improve the user experience without incurring a significant cost.

We formulate the problem of finding the optimal network association policy  as a Markov Decision Process (MDP) with a probabilistic constraint. Similarly to the Bellman equation proposed by Chen \cite{Chen04}, for a discrete-time MDP with probabilistic constraints, we write the Hamilton-Jacobi-Bellman (HJB) equation for our continuous-time problem by proper state space expansion. The HJB equation is instrumental in optimality verification of a particular control policy for which the expected cost is explicitly characterized. However, due to discontinuity of the queue-length process for the Poisson arrival model, a closed-form characterizations of our proposed policies are not available for the verification of the HJB equation. Therefore, we consider a fluid approximation model, where the arrival process is modeled as a controlled Brownian motion with a drift. In this case, we provide and exact and explicit characterization of a threshold policy that satisfies the QoE constraints. We show that the expected cost corresponding to this threshold policy indeed is the solution of the corresponding HJB equation, thus proving the optimality of such policies.

\subsection{Related Work}

The set of works related to this thesis spans several distinct areas of the literature.
One of the major difficulties in the literature is the notion of \emph{delay}, which greatly varies across different applications and time scales at which the system is modeled. The role of delay-related metrics has been extensive in the literature on Network Optimization and Control. Neely \cite{Neely06, Neely10} employs Lyapunov optimization techniques to study the delay analysis of stochastic networks, and its trade-offs with other utility functions. Other related works such as \cite{Srikant99}, \cite{Berry02} take the  flow-based optimization approach, also known as Network Utility Maximization (NUM), to maximize the \emph{delay-related} utility of the users. Closer to our work is the one by Hou and Kumar \cite{HouKumar10} that considers per-packet delay constraints and successful delivery ratio as user experience metrics. Such flow-based approaches are essentially operating at the steady state of the system, and fail to capture the end-user experience for delay-sensitive applications of interest.

Media streaming, particularly in the area of P2P networks, has attracted significant recent interest. For example, works such as  \cite{ZhoChi07,BonMas08,ZhaLuiChi_09,YinSri10} develop analytical models on the trade-off between the steady state probability of missing a block and buffer size under different block selection policies.  Unlike our model, they consider live streaming, e.g. video conferencing, with deterministic channels. However, we focus on content that is at least partially cached at multiple locations, and must be streamed over one or more unreliable channels.  Further, our analysis is on transient effects---we are interested in the first time that media playback is interrupted as a function of the initial amount of buffering. Also related to our work is \cite{KumAlt07}, which considers two possible wireless access methods (WiFi and UMTS) for file delivery, assuming particular throughput models for each access method.  In contrast to this work, packet arrivals are stochastic in our model, and our streaming application requires hard constraints on quality of user experience.

Another body of related work is the literature on constrained Markov decision processes. There are two main approaches to these problems. Altman \cite{Altman_book}, Piunovskiy \cite{Piunovskiy97, Piunovskiy98} and Feinberg and Shwartz \cite{Feinberg95} take a convex analytic approach leading to linear programs for obtaining the optimal policies. On the other hand, Chen \cite{Chen04}, Chen and Blankenship \cite{ChenBlankenship04}, Piunovskiy \cite{Piunovskiy00} use the more natural and straightforward Dynamic Programming approach to characterize all optimal policies. These works mainly focus on different variations of the discrete-time Markov decision processes. In this work, we take the dynamic programming approach for the control of a continuous-time Markovian process. Further, we employ stochastic calculus techniques used in treatment of stochastic control problems \cite{BrockettBook} to properly characterize the optimal control policies.

The rest of this paper is organized as follows. In Section \ref{sec:coding}, we discuss using network coding techniques to guarantee efficient streaming in a technology heterogeneous system. Section \ref{sec:model} describes the system model and QoE metrics for a media streaming scenario from cost-heterogeneous servers. In Section \ref{sec:policies}, we present and compare several server association policies. The dynamic programming approach for characterization of the optimal control policy is discussed in Section \ref{dp_sec}. We present the fluid approximation model and establish the optimality of an online threshold policy in Section \ref{fluid_sec}. Finally, Section \ref{sec:conclusion} provides a summary of the contributions of this paper with pointers for potential extensions in the future.

\section{Network Coding for Efficient Streaming in Technology-heterogeneous Systems}\label{sec:coding}

In this part, we study the problem of streaming a media file from multiple servers to a single receiver over unreliable communication channels. Each of the servers could be a wireless access point, base station, another peer, or any combination of the above. Such servers may operate under different protocols in different ranges of the spectrum such as WiFi (IEEE 802.11), WiMAX (IEEE 802.16), HSPA, EvDo, LTE, etc. We refer to such system as a technology-heterogeneous multi-server system. In this setup, the receiver can request different pieces of the media file from different servers.  Requesting packets form each server may cause delays due to channel uncertainty. However, requesting one packet from multiple servers introduces the need to keep track of packets, and the duplicate packet reception problem. In this section, we discuss methods that enable efficient streaming across different paths and network interfaces. This greatly simplifies the model when analyzing such systems.

In order to resolve such issues of the multi-server and technology-heterogeneous systems, let us take a closer look at the process of media streaming across different layers.  Media files are divided into blocks of relatively large size, each consisting of several frames.  The video coding is such that all the frames in the block need to be available before any frames can be played. Blocks are requested in sequence by the playback application from the user-end.  The server (or other peers) packetize the requested block and transmit them to the user as in Figure \ref{layers_fig}.

\begin{figure}[htbp]
\centering
  \includegraphics[width=4in]{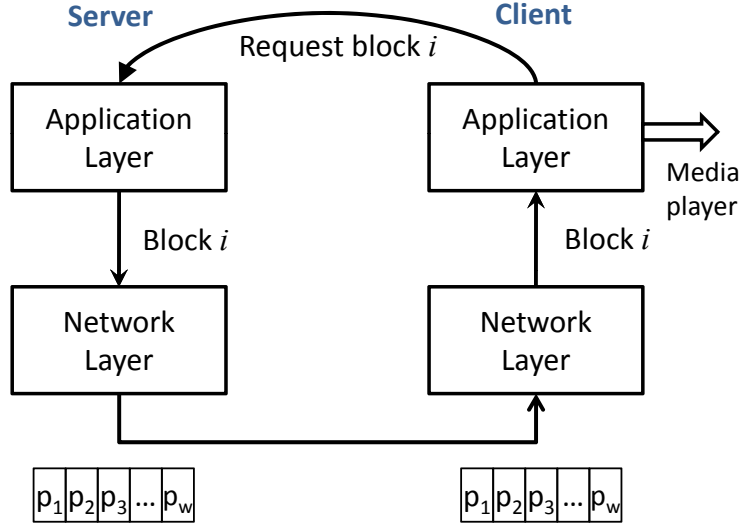}
  \caption{The media player (application layer) requires complete blocks. At the network layer each block is divided into packets and delivered. }\label{layers_fig}
\end{figure}

Now consider the scenario, illustrated in Figure \ref{multipath_fig}, where there are multiple paths to reach a particular server. Each of these paths could pass through different network infrastructures. For example, in Figure \ref{multipath_fig}, one of the paths is using the WiFi access point, while the other one is formed by the LTE (Long Term Evolution) network.

\begin{figure}[htbp]
\centering
  \includegraphics[width=5in]{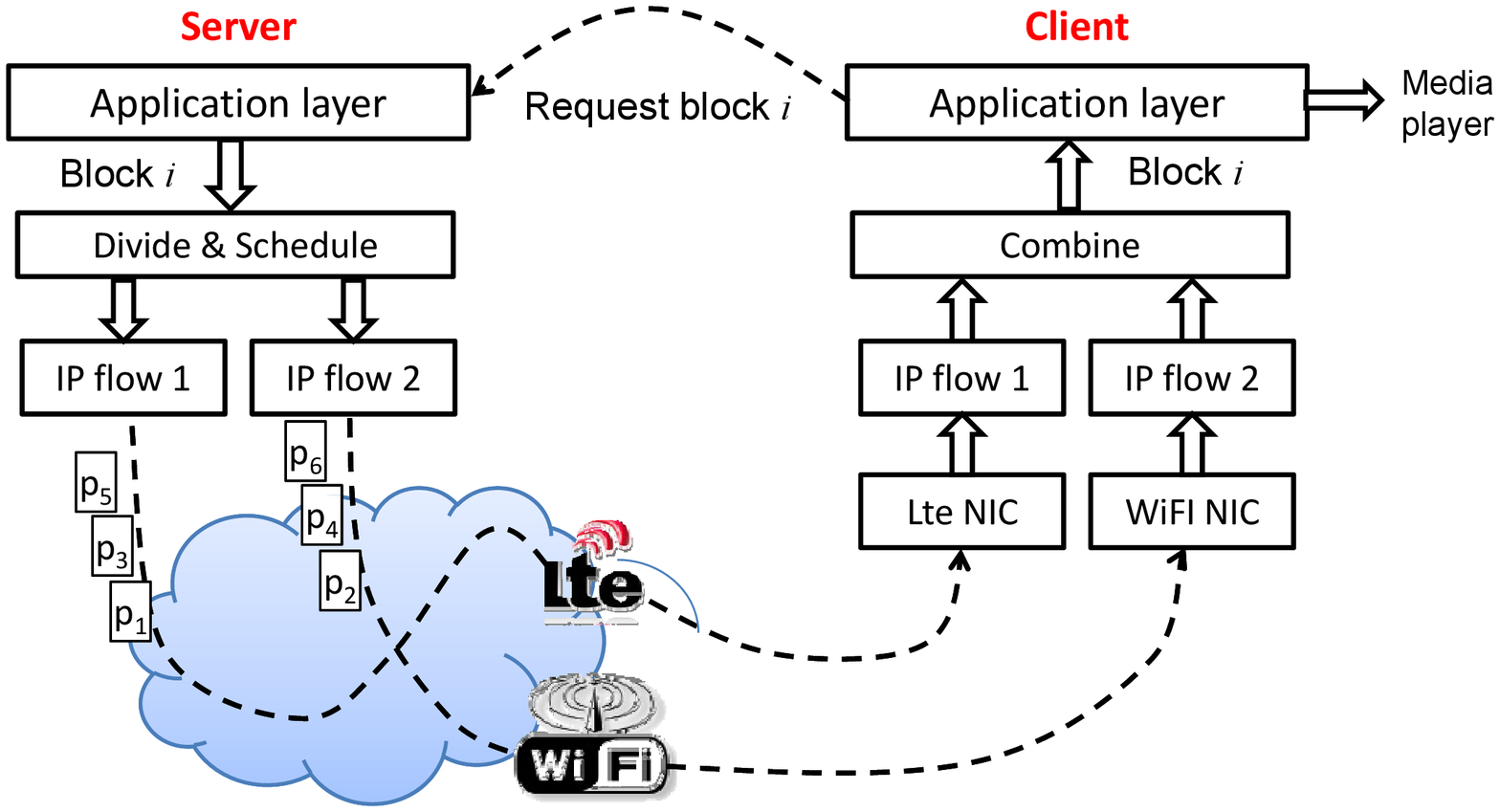}
  \caption{Streaming over multiple paths/interfaces.}\label{multipath_fig}
\end{figure}

The conventional approach in exploiting the path diversity in such scenarios is scheduling each packet to be delivered over one of the available paths. For instance, odd-numbered packets are assigned to path 1, and even-numbered packets are assigned to path 2. This approach requires a  deterministic and reliable setting, where each path is lossless and the capacity and end-to-end delay of each path is known. However, the wireless medium is intrinsically unreliable and time varying. Moreover, flow dynamics in other parts of the network may result in congestion on  a particular path. Therefore, the slowest or most unreliable path becomes the bottleneck. In order to compensate for that, the scheduler may add some redundancy by sending the same packet over multiple paths, which results in duplicate packet reception and loss of performance. There is a significant effort to use proper scheduling and control mechanisms to reduce these problems. For more information on this approach, generally known as MultiPath TCP (MPTCP), please refer to the works by Wischik \emph{et al.} \cite{MPTCP1, MPTCP2}, and IETF working draft on MPTCP \cite{MPTCP_IETF}.

\begin{figure}[htbp]
\centering
  \includegraphics[width=5in]{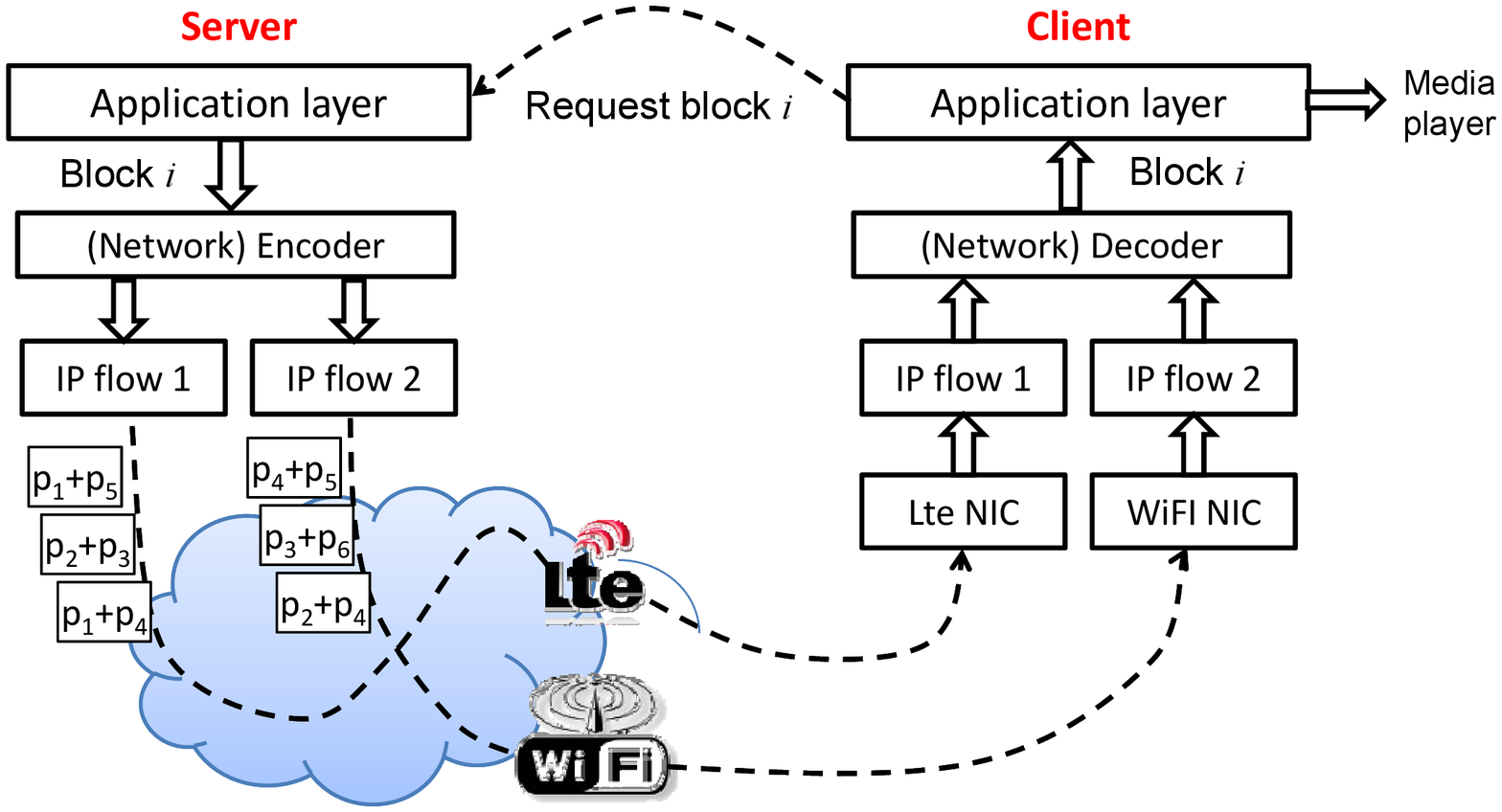}
  \caption{Streaming over multiple paths/interfaces using Network Coding.}\label{multipath_NC_fig}
\end{figure}

We propose random linear network coding (RLNC) to alleviate the duplicate packet reception problem. Figure \ref{multipath_NC_fig} illustrates an example. Here, instead of requesting a particular packet in block $i$, the receiver simply requests a random linear combination of all the packets in block $i$.  The coefficients of each combination are chosen uniformly at random from a Galois field of size $q$. The coded packets delivered to the receiver can be thought of as linear equations, where the unknowns are the original packets in block $i$. Block $i$ can be fully recovered by solving a system of linear equations if it is full rank. Note that we can embed the coding coefficients in the header of each coded packet so that the receiver can form the system of linear equations. For more implementation details, please refer to Sundararajan \emph{et al.} \cite{TCP_jaykumar}. It can be shown that if the field size $q$ is large enough, the received linear equations are linearly independent with very high probability \cite{RLNC}. Therefore, for recovering a block of $W$ packets, it is sufficient to receive $W$ linearly independent coded packets from different peers. Each received coded packet is likely to be independent of previous ones with probability $1-\delta(q),$ where $\delta(q)\rightarrow 0$ as $q \rightarrow \infty$.

By removing the notion of \emph{unique identity} assigned to each packet, network coding allows significant simplification of the scheduling and flow control tasks at the server. For instance, if one of the paths get congested or drops a few of the packets, the server may complete the block transfer by sending more coded packets over the other paths. Hence, the sender may perform TCP-like flow management and congestion control on each of the paths \emph{independently}. Therefore, network coding provides a mean to homogenize a technology-heterogeneous system as if the receiver only has single interface.

Note that such random linear coding does not introduce additional decoding delay for each block, since the frames in a block can only be played out when the whole block is received.   So there is no difference in delay whether the end-user received $W$ uncoded packets of the block or $W$ independent coded packets that can then be decoded.

In the following, we discuss the conditions under which we can convert a technology-heterogeneous multi-server system to a single-path single-server system. Consider a single user receiving a media file from various servers it is connected to. Assume that the media file is divided into blocks of $W$ packets. Each server sends  random linear combinations of the packets within the current block to the receiver. We assume that the linear combination coefficients are selected from a Galois field of large enough size, so that no redundant packet is delivered to the receiver. Moreover, we assume that the block size $W$ is small compared to the total length of the file, but large enough to ignore the boundary effects of moving from one block to the next.  For simplicity of the analysis, we assume in this work that the arrival process of packets from each server is a Poisson process. Since, network coding allows for independent flow control on each path, we may assume that the arrival process from each server is independent of other arrival processes. Moreover, since we can assume no redundant packet is delivered from different peers, we can combine the arrival processes into one Poisson process of the sum-rate $R$. Thus, our simplified model is just a single-server-single-receiver system.


We summarize the above discussions into the following Assumption, which is the key for development of the analytical results in the subsequent parts.

\begin{assumption}\label{prop:multipath_model}
Consider one or more servers streaming a single media file to a single client over $m$ independent path using random linear network coding. The packet delivery process over path $k$ is modeled as a Poisson process of rate $R_k$. The effective packet delivery process observed at the receiver is a Poisson process of rate $R = \sum_{k=1}^m R_k$.
\end{assumption}

Note that, for simplicity of the system model and analysis,  we are neglecting a few complexities associated with network coding approach such as the effect of field size, feedback imperfections, and other uncertainties

Assumption \ref{prop:multipath_model} provides the necessary tool for analyzing technology-heterogeneous multi-server systems as a single-server system. For instance, we can apply most of the results of \cite{JSAC} on fundamental delay-interruption trade-offs. This is essential for tractability of the analysis of cost-heterogeneous systems, which is the focus of the subsequent part.

\section{System Model and QoE Metrics}\label{sec:model}
We consider a media streaming system as follows. A single user is receiving a media stream of infinite size, from various servers or access points. The receiver first buffers $D$ packets from the beginning of the file, and then starts the playback at unit rate.

We assume that time is continuous, and the arrival process of packets from each server is a Poisson process independent of other arrival processes. Further, we assume that each server sends random linear combination of the packets in the source file. Therefore, by discussions of Section \ref{sec:coding}, no redundant packet is delivered from different servers. Therefore, we can combine the arrival processes of any subset of the servers into one Poisson process of rate equal to the sum of the rates from the corresponding servers (cf.  Assumption \ref{prop:multipath_model}).

There are two types of servers in the system: free \footnote{The contributions of this work still hold if both servers are costly with different access costs. Here, we normalize the access cost of the cheaper server to zero.} servers and the costly ones. There is no cost associated with receiving packets from a free server, but a unit cost is incurred per unit time the  costly servers are used. As described above, we can combine all the free servers into one free server from which packets arrive according to a Poisson process of rate $R_0$. Similarly, we can merge all of the costly servers into one costly server with effective rate of $R_c$. At any time $t$, the user has the option to use packets only from the free server or from both the free and the costly servers. In the latter case, the packets arrive according to a Poisson process of rate $$R_1 = R_0 + R_c.$$  The user's action at time $t$ is denoted by $u_t \in \{ 0,1 \}$, where $u_t =0$ if only the free server is used at time $t,$ while $u_t =1$ if both free and costly servers are used. We assume that the parameters $R_0$ and $R_1$ are known at the receiver. Figure \ref{MultiServer_model_fig} illustrates the system model.

\begin{figure}[htbp]
\centering
  \includegraphics[width=3in]{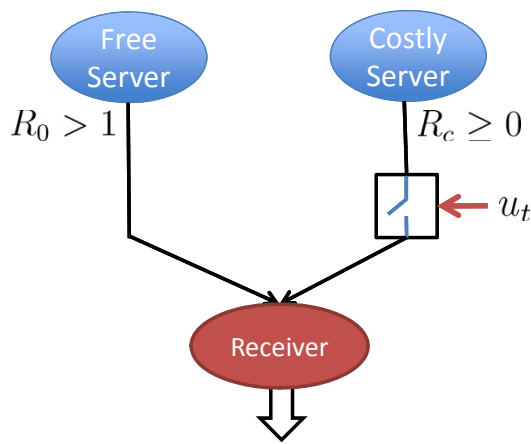}
  \caption{Streaming frow two classes of servers: costly and free.}\label{MultiServer_model_fig}
\end{figure}

 The dynamics of the receiver's buffer size (queue-length) $Q_t$ can be described as follows
\begin{equation}\label{buffer_het}
    Q_t =  D + N_t + \int_0^t u_\tau dN^c_\tau  - t,
\end{equation}
where $D$ is the initial buffer size, $N_t$ Poisson processes of rate $R_0$ and $N^c_t$ is a Poisson counter of rate $R_c$ which is independent of the process $N_t$. The last term correspond to the unit rate of media playback.

The user's association (control\footnote{Throughout the rest of this paper, we use the notion of control policy and association policy, interchangeably.}) policy is formally defined below.
\begin{definition}\label{policy_def}
[Control Policy] Let $$h_t = \{Q_s: 0\leq s\ \leq t\} \cup \{u_s: 0\leq s\ < t\}$$ denote the history of the buffer sizes and actions up to time $t$, and $\mathcal H$ be the set of all histories for all $t$. A \emph{deterministic association policy } denoted by $\pi$ is a mapping $\pi : \mathcal H \longmapsto \{ 0,1 \} $, where at any time $t$
$$u_t = \pi(h_t) = \left\{
             \begin{array}{ll}
               0, & \hbox{if only the free server is chosen,} \\
               1, & \hbox{if both servers are chosen.}
             \end{array}
           \right.
$$
Denote by $\Pi$ the set of all such control policies.
\end{definition}

We use the initial buffer size $D$, and interruption probability as QoE metrics. The interruption event occurs when the queue-length $Q_t$ reaches zero.
However, this event not only depends on the initial buffer size $D$, but also on the control policy $\pi$. To emphasize this dependency, we denote the interruption probability by
\begin{equation}\label{p_int_pi}
p^{\pi}(D) = \pr \{\tau_0 < \infty\},
\end{equation}
where
$$\tau_0 = \inf\{t: Q_t = 0\}.$$

\begin{definition}\label{feas_policy_def}
The policy $\pi$ is defined to be  $(D,\epsilon)$-\emph{feasible} if $p^{\pi}(D) \leq \epsilon$. The set of all such feasible policies is denoted by $\Pi(D,\epsilon)$.
\end{definition}

The third metric that we consider in this work is the expected cost of using the costly server which is proportional to the expected usage time of the costly server. For any $(D,\epsilon)$, the usage cost of a $(D,\epsilon)$-feasible policy $\pi$ is given by\footnote{Throughout this work, we use the convention that the cost of an infeasible policy is infinite.}
\begin{equation}\label{cost_def}
J^{\pi}(D,\e) = \E \Big[\int_0^{\tau_0} u_t dt\Big].
\end{equation}

The \emph{value function} or optimal cost function $V$ is defined as
\begin{equation}\label{value_def}
    V(D, \e) = \min_{\pi \in \Pi(D,\epsilon)} J^{\pi}(D,\e),
\end{equation}
and the optimal policy $\pi^*$ is defined as the optimal solution of the minimization problem in (\ref{value_def}).

In our model, the user expects to buffer no more than $D$ packets and have an interruption-free experience with probability higher than a desired level $1-\epsilon$. Note that there are trade-offs  among the interruption probability $\epsilon$, the initial buffer size $D$, and the usage cost. These trade-offs depend on the association policy as well as the system parameters $R_0$, $R_c$ and $F$.

Throughout the rest of this work, we study the case that $R_0 > 1$ and the file size $F$ goes to infinity,  since the control policies in this case take simpler forms. Moreover, the cost of such control policies  in this case provide an upper bound for the finite file size case. The following Lemma summarizes the main result from \cite{JSAC}, characterizing the fundamental trade-off between the interruption probability and initial buffering, for a single-server single-receiver system.

\begin{lemma}\label{pd_exact_lemma}
Consider a single receiver receiving a media stream from a single server according to a Poisson process of rate $R$. Let $D$ denote the initial buffer size before the playback, and set the playback rate to one. The  probability of interruption in media playback is given by
\begin{equation}\label{pd_exact}
 p(D) = e^{-I(R) D},
\end{equation}
where $I(R)$ is the largest root of $\gamma(r) = r + R(e^{-r} -1)$.
\end{lemma}

We first characterize the region of interest in the space of QoE metrics. In this region, a feasible control policy exists and is non-degenerate. We then use these results to design proper association policies.

\begin{theorem}\label{boundary_thm}
Let $(D,\e)$ be user's QoE requirement when streaming an infinite file from two servers. The arrival rate of the free server is given by $R_0 > 1$, and the total arrival rate when using the costly server is denoted by $R_1 > R_0$. Then

(a) For any $(D,\epsilon)$ such that $D\geq \frac{1}{I(R_0)}\log\big(\frac1\epsilon\big)$,
$$\min_{\pi \in \Pi} J^{\pi}(D,\e)  = 0.$$

(b) For any $(D,\epsilon)$ such that $D < \frac{1}{I(R_1)}\log\big(\frac1\epsilon\big)$,
$$\min_{\pi \in \Pi} J^{\pi}(D,\e) = \infty.$$

\end{theorem}
\begin{proof}
Consider the degenerate policy $\pi_0 \equiv 0$. This policy is equivalent to a single-server system with arrival rate $R = R_0$. By Lemma \ref{pd_exact_lemma}, the policy $\pi_0$ is $(D,\epsilon)$-feasible for all $D\geq \frac{1}{I(R_0)}\log\big(\frac1\epsilon\big)$. Note that by (\ref{cost_def}) this policy does not incur any cost, which results in part (a).

Moreover, for all $(D,\epsilon)$ with $D < \frac{1}{I(R_1)}\log\big(\frac1\epsilon\big)$, there is no $(D,\epsilon)$-feasible policy. This is so since the buffer size under any policy $\pi$ is stochastically dominated by the one governed by the degenerate policy  $\pi_1 \equiv 1$. Hence,
$$p^{\pi}(D) \geq p^{\pi_1}(D) = \exp(-I(R_1) D) > \epsilon.$$
Using the convention of infinite cost for infeasible policies, we obtain the result in part (b).
\end{proof}

\begin{figure}[htbp]
\centering
  \includegraphics[width=3.5in]{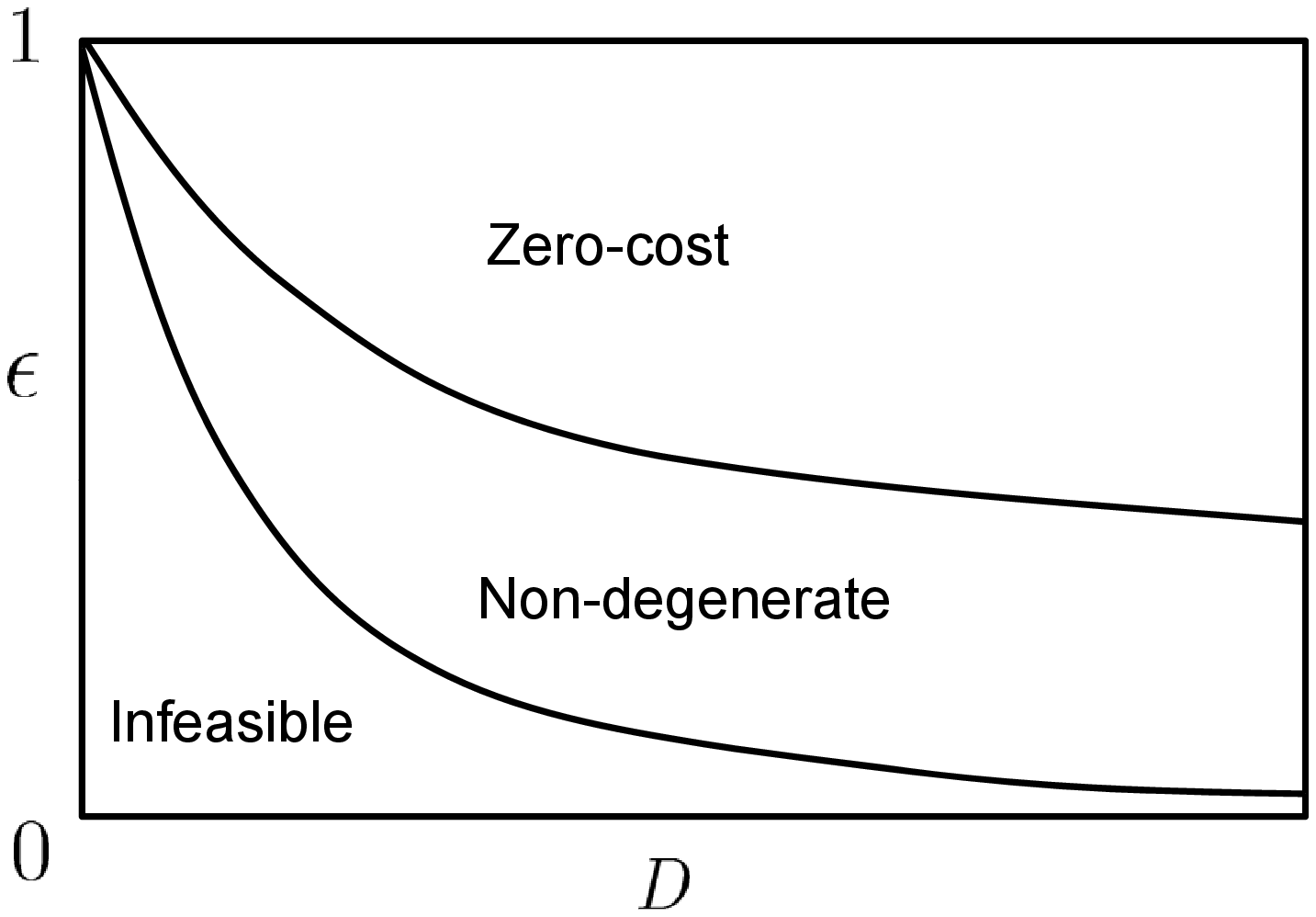}
  \caption{Non-degenerate, zero-cost and infeasible regions for QoE metrics $(D,\e)$.}\label{regions_fig}
\end{figure}

For simplicity of notation, let $\alpha_0 = I(R_0)$, and $\alpha_1 = I(R_1)$.  By Theorem \ref{boundary_thm} we focus on the region
\begin{equation}\label{S_def}
\mathcal R = \Big\{(D,\epsilon): \frac{1}{\alpha_1}\log\big(\frac1\epsilon\big) \leq D \leq \frac{1}{\alpha_0}\log\big(\frac1\epsilon\big)\Big\}
\end{equation}
to analyze the expected cost of various classes of control policies. Figure \ref{regions_fig} illustrates a conceptual example of this non-degenerate region as well as the zero-cost and infeasible regions.

\section{Design and Analysis of  Association Policies}\label{sec:policies}
In this part, we propose several classes of parameterized control policies. We first characterize the range of the parameters for which the association policy is feasible for a given initial buffer size $D$ and the desired level of interruption probability $\epsilon$. Then, we try to choose the parameters such that the expected cost of the policy is minimized. All of the proofs of the main theorems are included in Appendix \ref{app:policy_analysis}.

\subsection{Off-line Policy}
Consider the class of policies where the decisions are made off-line before starting the media streaming. In this case, the arrival process is not observable by the decision maker. Therefore, the user's decision space reduces to the set of deterministic functions $u: \R \rightarrow \{0,1\}$, that maps time into the action space.

\begin{theorem}\label{offline_form_thm}
Let the cost of a control policy be defined as in (\ref{cost_def}). In order to find a minimum-cost off-line policy, it is sufficient to consider policies of the form:
\begin{equation}\label{offline_form}
    \pi(h_t) = u_t = \left\{
            \begin{array}{ll}
              1, & \hbox{if $t \leq t_s$} \\
              0, & \hbox{if $t > t_s,$}
            \end{array}
          \right.
\end{equation}
which parameterized are by a single parameter $t_s \geq 0$.
\end{theorem}

\begin{proof}
In general, any off-line policy $\pi$ consists of multiple intervals in which the costly server is used. Consider an alternative policy $\pi'$  of the form of (\ref{offline_form}) where $t_s = J^{\pi}$. By definition of the cost function in (\ref{cost_def}), the two policies incur the same cost. Moreover, the buffer size process under policy $\pi$ is stochastically dominated by the one under policy $\pi'$, because the policy $\pi'$ counts the arrivals from the costly server earlier, and the arrival process is stationary. Hence, the interruption probability of $\pi'$ is not larger than that of $\pi$. Therefore, for any off-line policy, there exists another off-line policy of the form given by (\ref{offline_form}).
\end{proof}

\begin{theorem}\label{offline_range_thm}
Consider the class of off-lines policies of the form (\ref{offline_form}). For any $(D,\e) \in \mathcal R$, the policy $\pi$ defined in (\ref{offline_form}) is feasible if
\begin{equation}\label{ts_range}
    t_s \geq t_s^* = \frac{R_0}{R_1 - R_0} \bigg[ \frac{1}{\alpha_0} \log\Big( \frac{1}{\e - e^{-\alpha_1 D}}  \Big) - D \bigg].
\end{equation}
\end{theorem}

Note that obtaining the optimal off-line policy is equivalent to finding the smallest $t_s$ for which the policy is still feasible. Therefore, $t_s^*$ given in (\ref{ts_range}) provides an upper bound on the minimum cost of an off-line policy. Observe that $t^*_s$ is almost linear in $D$ for all $(D,\e)$ that is not too close to the lower boundary of region $\mathcal R$. As $(D,\e)$ gets closer to the boundary, $t^*_s$ and the expected cost grows to infinity, which is in agreement with Theorem \ref{boundary_thm}. In this work, we pick $t^*_s$ as a benchmark for comparison to other policies that we  present next.

\subsection{Online Safe Policy}
Let us now consider the class of online policies where the decision maker can observe the buffer size history. Inspired by the structure of the optimal off-line policies, we first focus on a \emph{safe} control policy in which, in order to avoid interruptions, the costly server is used  at the beginning  until the buffer size reaches a certain threshold, after which the costly server is never used. This policy is formally defined below.

\begin{definition}\label{safe_def}
The online safe policy $\pi^S$ parameterized by the threshold value $S$ is given by
\begin{equation}\label{safe_policy}
    \pi^S(h_t) = \left\{
                   \begin{array}{ll}
                     1, & \hbox{if $t \leq \tau_s$} \\
                     0, & \hbox{if $t > \tau_s$,}
                   \end{array}
                 \right.
\end{equation}
where $\tau_S = \inf \{t \geq 0: Q_t \geq S\}$.
\end{definition}

\begin{theorem}\label{safe_thm}
Let $\pi^S$ be the safe policy defined in Definition \ref{safe_def}.
For any $(D,\e) \in \mathcal R$, the safe policy is feasible if
\begin{equation}\label{S_range}
    S \geq S^* = \frac{1}{\alpha_0} \log\Big( \frac{1}{\e - e^{-\alpha_1 D}}  \Big).
\end{equation}
Moreover,
\begin{eqnarray*}\label{safe_cost}
    \min_{S\geq S^*} J^{\pi^S}(D,\e)\!\!\!\! &=& \!\!\!\! J^{\pi^{S^*}}(D,\e)
   =  \frac{1}{R_1 - 1} \bigg[ \frac{1}{\alpha_0} \log\Big( \frac{1}{\e - e^{-\alpha_1 D}}  \Big) - D +\xi\bigg],
\end{eqnarray*}
where $\xi \in [0,1)$.
\end{theorem}

Let us now compare the online safe policy $\pi^{S^*}$ with the off-line policy defined in (\ref{offline_form}) with parameter $t_s^*$ as in (\ref{ts_range}). We observe that the cost of the online safe policy is almost proportional to that of the off-line policy, where the cost ratio of the off-line policy to that of the online safe policy is given by
$$\frac{R_0(R_1-1)}{R_1-R_0} = 1 + \frac{R_1(R_0-1)}{R_1-R_0} > 1.$$
Note that the structure of both policies is the same, i.e, both policies use the costly server for a certain period of time and then switch back to the free server. As suggested here, the advantage of observing the buffer size allows the online policies to avoid excessive use of the costly server when there are sufficiently large number of arrivals from the free server. In the following, we present another class of online policies.

\subsection{Online Risky Policy}
In this part, we study a class of online policies where the costly server is used only if the buffer size is below a certain threshold. We call such policies ``risky'' as the risk of interruption is spread out across the whole trajectory, unlike the ``safe'' policies.  Further, we constrain risky policies to possess the property that the action at a particular time should only depend on the buffer size at that time, i.e., such policies are \emph{stationary Markov} with respect to buffer size as the state of the system. The risky policy is formally defined below.

\begin{definition}\label{risky_def}
The online risky policy $\pi^T$ parameterized by the threshold value $T$ is given by
\begin{equation}\label{risky_policy}
    \pi^T(h_t) = \pi^T(Q_t)= \left\{
                   \begin{array}{ll}
                     1, & \hbox{if $0< Q_t  <  T$} \\
                     0, & \hbox{otherwise.}
                   \end{array}
                 \right.
\end{equation}
\end{definition}

\begin{theorem}\label{risky_thm}
Let $\pi^T$ be the risky policy defined in Definition \ref{risky_def}.
For any $(D,\e) \in \mathcal R$, the policy $\pi^T$ is feasible if the threshold $T$ satisfies
\begin{equation}\label{threshold}
    T \geq T^* = \left\{
                   \begin{array}{ll}
                     \frac{1}{\a _1-\a _0}\big[\log\big(\frac{\beta}{\e}\big) - \a_0 D \big], & \hbox{if $D \geq \bar D$,} \\
                     \frac{1}{\a_1}\log\Big(\frac{\e + \beta(1-e^{-\a_1 D}) - 1}{\e - e^{-\a_1 D}}\Big), & \hbox{if $D \leq \bar D$,}
                   \end{array}
                 \right.
\end{equation}
where $\beta = \frac{\a_1}{\a_0 (1-\frac{\a_0}{2})}$ and $\bar D = \frac{1}{\a_1}\log\big(\frac{\beta}{\e}\big)$.

\end{theorem}

Theorem \ref{risky_thm} facilitates the design of risky policies with a single-threshold structure, for any desired initial buffer size $D$ and interruption probability $\e$. For a fixed $\e$, when $D$ increases, $T^*$ (the design given by Theorem \ref{risky_thm}) decreases to zero. On the other hand, if $D$ decreases to $\frac{1}{\a_1}\log\big(\frac1\e\big)$ (the boundary of $\mathcal R$), the threshold $T^*$ quickly increases to infinity, i.e., the policy does not switch back to the free server unless a sufficiently large number of packets is buffered. Figure \ref{T*D_figure} plots $T^*$ and $D$ as a function of $D$ for a fixed $\e$. Observe that, for large range of $D$, $T^* \leq D$, i.e., the costly server is not initially used. In this range, owing to the positive drift of $Q_t$, the probability of ever using the costly server exponentially decreases in $(D - T^*)$.

\begin{figure}[htbp]
\centering
  \includegraphics[width=4in]{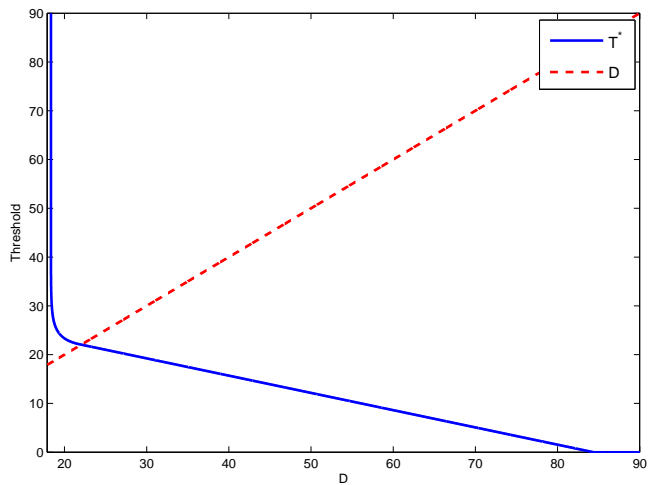}
  \caption{The switching threshold of the online risky policy as a function of the initial buffer size for $\e = 10^{-3}$ (See Theorem \ref{risky_thm}).}\label{T*D_figure}
\end{figure}

Next, we compute bounds on the expected cost of the online risky policy and compare with the previously proposed policies.

\begin{theorem}\label{risky_cost_thm}
For any $(D, \e) \in \mathcal R$, consider an online risky policy $\pi^{T^*}$ defined in Definition \ref{risky_def}, where the threshold $T^*$ is given by (\ref{threshold}) as function of $D$ and $\e$. If $D \geq \bar D$ then
\begin{equation}\label{risky_cost1}
  J^{\pi^{T^*}}(D,\e) \leq    \frac{\beta}{\a_1(R_1-1)}e^{-a_0(D- T^*)},
\end{equation}
and if $D \leq \bar D$
\begin{equation}\label{risky_cost2}
  J^{\pi^{T^*}}(D,\e) \leq \frac{1 - e^{-\a_1 D}}{(R_1-1)(1 - e^{-\a_1 T^*})}\Big(T^*+1+\frac{\beta}{\a_1}\Big) - \frac{D}{R_1-1},
\end{equation}
where $\beta = \frac{\a_1}{\a_0 (1-\frac{\a_0}{2})}$ and $\bar D = \frac{1}{\a_1}\log\big(\frac{\beta}{\e}\big)$.
\end{theorem}

In the following, we compare the expected cost of the presented policies using numerical methods, and illustrate that the bounds  derived in Theorems  \ref{offline_range_thm}, \ref{safe_thm} and \ref{risky_cost_thm} on the expected cost function are close to the exact value.

\subsection{Performance Comparison}

Figure \ref{cost_comparison_fig} compares the expected cost functions of the off-line, online safe and online risky policies as a function of the initial buffer size $D$, when the interruption probability is fixed to $\e = 10^{-3}$, the arrival rate from the free server is $R_0 = 1.05$, and the arrival rate from the costly server is $R_c = R_1-R_0 = 0.15$. We plot the bounds on the expected cost given by Theorems \ref{offline_range_thm}, \ref{safe_thm} and \ref{risky_cost_thm} as well as the expected cost function numerically computed by the Monte-Carlo method. Figure \ref{cost_comparison_fig} shows that the analytical bounds we computed for the expected cost of various control policies closely match the exact cost functions computed via simulations.

\begin{figure}[htbp]
\centering
  \includegraphics[width=4in]{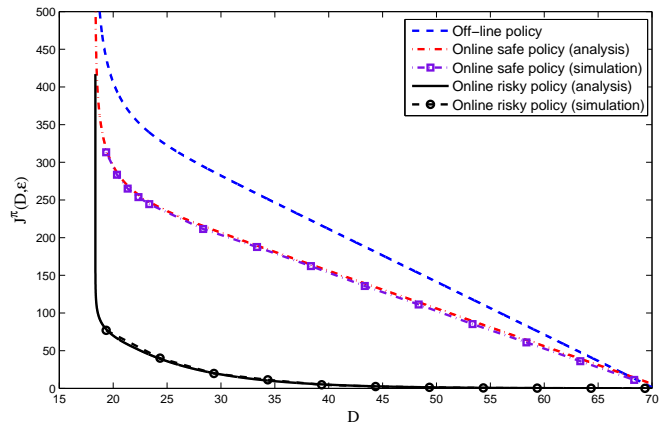}
  \caption{Expected cost (units of time) of the presented control policies as a function of the initial buffer size for interruption probability $\e = 10^{-3}$. The analytical bounds are given by Theorems \ref{offline_range_thm}, \ref{safe_thm} and \ref{risky_cost_thm}.}\label{cost_comparison_fig}
\end{figure}

Observe that the expected cost of the risky policy is significantly smaller that both online safe and off-line policies. For example, the risky policy allows us to decrease the initial buffer size from 70 to 20 with an average of $70 \times 0.15 \approx 10$ extra packets from the costly server.  The expected cost in terms of the number packets received from the costly server is 43 and 61 for the online safe and off-line policy, respectively.

Moreover, note that it is merely the existence of the costly server as a backup that allows us to improve the user's quality of experience without actually using too many packets from the costly server. For example, observe that the risky policy satisfies QoE metrics of $D = 35$ and $\e = 10^{-3}$, by only using on average about one extra packet from the costly server. However, without the costly server, in order to decrease the initial buffer size from 70 to 35, the interruption probability has to increase from $10^{-3}$ to about $0.03$ (see Lemma \ref{pd_exact_lemma}).

\section{Dynamic Programming Approach}\label{dp_sec}
In this section, we present a characterization of the optimal association policy in terms of the Hamilton-Jacobi-Bellman (HJB) equation. Note that because of the probabilistic constraint over the space of sample paths of the buffer size, the optimal policy is not necessarily Markov with respect to the buffer size as the state of the system. We take a similar approach as in \cite{Chen04} where by expanding the state space, a Bellman equation is provided as the optimality condition of an MDP with probabilistic constraint. In particular, consider the pair $(Q_t,p_t)$ as the state variable, where $Q_t$ denotes the buffer size and $p_t$ represents the desired level of interruption probability given the information at time $t$. Note that $p_t$ is a Martingale by definition \cite{SP_book}.
The evolution of $Q_t$ is governed by the following stochastic differential equation
\begin{equation}\label{x_SDE}
    dQ_t = -dt + dN^u,\quad Q_0 = D,
\end{equation}
where $N^u$ is a Poisson counter with rate $R_{u_t} = R_0 + u_t\cdot R_c$. For any $(D,\e) \in \mathcal R$ and any optimal policy $\pi$, the constraint $p^{\pi}(D) \leq \e$ is active. Hence, we consider the sample paths of $p_t$ such that $p_0 = \e$. Moreover, we have $\E[p_t] = \e$ for all $t$, where the expectation is with respect to the Poisson jumps. Let $dp_t = \p_t - p_t$ be the change in state $p$, if a Poisson jump occurs in an infinitesimal interval of length $dt$. Also, let $dp_t = dp_0$ be the change in state $p$ if no jump occurs. Therefore,
$$0 = \E[dp_t] = R_{u_t} dt (\p_t - p_t) + (1-R_{u_t} dt)dp_0.$$
By solving the above equation for $dp_0$, we obtain the evolution of $p_t$ as a function of the control process $\p_t$ and $u_t$:
\begin{equation}\label{p_SDE}
    dp_t = (p_t - \p_t) (R_{u_t} dt - dN^u), \quad p_0 = \e.
\end{equation}

Similarly to the arguments of Theorem 2 of \cite{Chen04}, by principle of optimality we can write the following dynamic programming equation
\begin{equation}\label{HJB_eq1}
    V(Q, p) = \min_{u \in\{0,1\}, \p \in[0,1]} \big\{ u dt +\E[V(Q+dQ, p+dp)] \big\}.
\end{equation}

If $V$ is continuously differentiable, by It\={o}'s Lemma for jump processes, we have
\begin{eqnarray*}
V(Q+dQ, p+dp) - V(Q,p) &=& \frac{\partial V}{\partial Q} (-dt) + \frac{\partial V}{\partial p} \cdot (p - \p)R_u dt  + \big(V(Q+1,\p) - V(Q,p)\big)dN^u,
\end{eqnarray*}
which implies the following HJB equation after dividing (\ref{HJB_eq1}) by $dt$ and taking the limit as $t$ goes to zero:
\begin{eqnarray}\label{HJB}
   \frac{\partial V (Q,p)}{\partial Q} &=&  \min_{u \in\{0,1\}, \p \in[0,1]} \big\{ u+ \frac{\partial V}{\partial p} \cdot (p - \p)R_u \nonumber  + R_u \big(V(Q+1,\p) - V(Q,p)\big)   \big\}
\end{eqnarray}

The optimal policy $\pi$ is obtained by characterizing the optimal solution of the partial differential equation in (\ref{HJB}) together with the boundary condition $V(Q,1) = 0$. Since such equations are in general difficult to solve analytically, we use the \emph{guess and check} approach, where we propose a candidate for the value function and verify that it nearly satisfies the HJB equation almost everywhere.
For any $(Q,p) \in \mathcal R$, define
\begin{equation}\label{T_xp}
    T(Q,p) = \left\{
                   \begin{array}{ll}
                     \frac{1}{\a _1-\a _0}\big[\log\big(\frac{\theta}{p}\big) - \a_0 Q \big], & \hbox{if $Q \geq  \frac{1}{a_1}\log\big(\frac{\theta}{p}\big)$,} \\
                     \frac{1}{\a_1}\log\Big(\frac{p + \theta(1-e^{-\a_1 Q}) - 1}{p - e^{-\a_1 Q}}\Big), & \hbox{otherwise,}
                   \end{array}
                 \right.
\end{equation}
where $\theta = \frac{\a_1}{\a_0}$. The candidate solution for HJB equation (\ref{HJB}) is given by
\begin{equation}\label{candidate1}
    \bar V(Q,p) = \frac{1}{\a_0(1-\frac{\a_0}{2})(R_1-1)}e^{-a_0(Q- T(Q,p))},
\end{equation}
when $Q \geq  \frac{1}{a_1}\log\big(\frac{\theta}{p}\big)$, and
\begin{equation}\label{candidate2}
  \bar V(Q,p) = \frac{p + \theta(1-e^{-\a_1 Q}) - 1}{(R_1-1)(\theta -1)}\big(T(Q,p)+\frac{\beta}{\a_1}\big) - \frac{Q}{R_1-1},
\end{equation}
when $Q < \frac{1}{a_1}\log\big(\frac{\theta}{p}\big)$. Note that the candidate solution is derived  from the structure of the expected cost of the risky policy (cf. Theorem \ref{risky_cost_thm}). We may verify that $\bar V$ satisfies the HJB equation (\ref{HJB}) for all $(Q,p)$ such that $Q \geq  \frac{1}{a_1}\log\big(\frac{\theta}{p}\big)$ or $Q \leq  \frac{1}{a_1}\log\big(\frac{\theta}{p}\big)-1$, but for other $(Q,p)$ the HJB equation is only approximately satisfied. This is due to the discontinuity of the queue-length process which does not allow us to exactly match the expected cost starting from below the threshold with the one starting from above the threshold.
Therefore, owing to approximate characterization of the cost of the risky policy, we may not prove or disprove optimality of this policy.

In the following, we use a fluid model to provide an \emph{exact} characterization of the optimal control policy using appropriate HJB equation. We show that the optimal policy takes a threshold structure similarly to the online risky policy.

\section{Optimal Association Policy for a Fluid Model}\label{fluid_sec}
Thus far, we concentrated on design and analysis of various network association policies in an uncertain environment, where network uncertainties are modeled using a Poisson arrival process. We provided closed-form approximations of the cost of different policies. However, an exact analytical solution is required to prove optimality of the risky policy.  This is particularly challenging, since an exact distribution of threshold over-shoots is desired, owing to the discontinuous nature of the Poisson process. In this part, we exploit a  second-order approximation of the Poisson process \cite{Kurtz78} and model the receiver's buffer size using a controlled Brownian motion with drift.

Consider the system model as in Figure \ref{MultiServer_model_fig}, with following queue-length dynamics at the receiver:
\begin{equation}\label{brownian_model}
    dQ_t = (R_{u_t} - 1) dt + dW_t, \quad Q_0 = D,
\end{equation}
where $W_t$ is the Wiener process, $u_t \in \{0, 1\}$ is the receiver's decision at time $t$ on using the free or costly server. As in the preceding part, we assume that the media file size $F$ is infinite and $R_1 > R_0 > 1$.

Define the control policy (network association policy) as in Definition \ref{policy_def}. The goal is to find a feasible policy that minimizes the usage cost defined in (\ref{cost_def}) such that the interruption probability $p^{\pi}(D)$ defined in (\ref{p_int_pi}) is at most $\e$. As in the previous part, the set of feasible policies and the value function is given by Definition \ref{feas_policy_def} and (\ref{value_def}), respectively. The following lemma is the counterpart of Lemma \ref{pd_exact_lemma} for the fluid model.

\begin{lemma}\label{lemma:degenerate}
Let $\pi_i \equiv i$ denote a degenerate policy, for $i \in \{0, 1\}$. The interruption probability for such policy is given by
\begin{equation}\label{pint_degenerate}
    p^{\pi_i}(D) = e^{-\t_i D}, \quad \foral D \geq 0,
\end{equation}
where $\t_i = 2(R_i - 1)$, for $i \in \{0,1\}$.
\end{lemma}
\begin{proof}
See Appendix \ref{app:fluid_policy_analysis}.
\end{proof}

First, we provide a characterization of the optimal policy via a Hamilton-Jacobi-Bellman equation. As in Section \ref{dp_sec}, we expand the state variables to  $(Q,p)$, where $Q$ is the queue-length with dynamics given by \ref{brownian_model}, and $p$ is the desired interruption probability. Using Martingale representation theorem \cite{SP_book}, we may write the dynamics of $p$ as follows:
\begin{equation}\label{p_dynamics}
    dp_t = \ph_t \ dW_t, \quad p_0 = \e,
\end{equation}
where $\ph_t$ is a predictable process which is adapted with respect to natural filtration of the history process. In this work, we focus on the control processes that are Markovian with respect to the state process $(Q_t, p_t)$. Therefore, using the principal of optimality, we may write the following dynamic programming equation:
\begin{equation}\label{HJB_fluid1}
    V(Q, p) = \min_{(u, \ph) \in\{0,1\}\times \R} \big\{ u dt +\E[V(Q+dQ, p+dp)] \big\},
\end{equation}
where $(u, \ph)$ are the control actions. For a twice differentiable function $V$, we may exploit It\={o}'s Lemma to get
\begin{eqnarray*}
dV(Q,p) &=&  \frac{\partial V}{\partial Q} dQ + \frac{\partial V}{\partial p} dp  +  \frac12 \frac{\partial^2 V}{\partial Q^2} (dQ)^2 + \frac12 \frac{\partial^2 V}{\partial p^2} (dp)^2 + \frac{\partial^2 V}{\partial Q \partial p} dQ dp \\
&\stackrel{(a)}{=}&  \frac{\partial V}{\partial Q} ((R_{u} - 1) dt + dW) +\frac{\partial V}{\partial p} (\ph \ dW) \\
&& +  \frac12 \frac{\partial^2 V}{\partial Q^2} dt + \frac12 \frac{\partial^2 V}{\partial p^2} (\ph)^2 dt  + \frac{\partial^2 V}{\partial Q \partial p} \ph dt,
\end{eqnarray*}
where $(a)$ follows from state dynamics in (\ref{brownian_model}) and (\ref{p_dynamics}). Replacing the above equation back in (\ref{HJB_fluid1}), and taking the expectation with respect to $dW$, and limit as $dt$ tends to zero, we obtain the following HJB equation
\begin{equation}\label{HJB_fluid2}
    0 = \min_{(u, \ph) \in\{0,1\}\times \R} \bigg\{ u +  \frac{\partial V}{\partial Q} (R_{u} - 1) + \frac12 \frac{\partial^2 V}{\partial Q^2} + \frac12 \frac{\partial^2 V}{\partial p^2} (\ph)^2 + \frac{\partial^2 V}{\partial Q \partial p} \ph \bigg\}.
\end{equation}

Note that we require the following boundary conditions for the value function:
\begin{equation}\label{BC_valuefunction}
    V(Q, 1) = V(0, p) = 0, \quad \foral Q \geq 0, 0 \leq p \leq 1
\end{equation}

Providing an analytical solution for the partial differential equation in (\ref{HJB_fluid2}) is often challenging. However, we may take a guess and check approach and use a threshold policy as the basis of our guess. Note that we need to verify the HJB equation for the set of state variables that reachable by a feasible policy. In particular, in light of Lemma \ref{lemma:degenerate}, it is clear that for $p \leq e^{-\t_1 Q}$, there is no feasible policy and the value function $V(Q,p) = \infty$. Moreover, for all $p \geq e^{-\t_0 Q}$, observe that the degenerate policy $\pi_0 \equiv 0$ is optimal, $V(Q,p) = 0$ which also satisfies the HJB equation and the boundary conditions. Therefore, we focus on the non-degenerate region
\begin{equation}\label{Region_def}
\mathcal R = \Big\{(Q,p): Q\geq 0, e^{-\t_1 Q} < p < e^{-\t_0 Q}  \Big\}.
\end{equation}

Figure \ref{regions_fig} illustrates a conceptual example of this non-degenerate region. In the following, we first define a threshold policy similar to the risky policy of Definition \ref{risky_def}, and present a closed-form characterization of its cost function. Then, we show that, for a proper choice of the threshold the associated cost function satisfies the HJB equation in (\ref{HJB_fluid2}), and the optimal solution of the minimization problem in (\ref{HJB_fluid2}) coincides with the threshold policy. Hence, we establish the optimality of the proposed policy.

\begin{theorem}\label{thm:threshold_policy}
Let $\pi^T$ be the threshold policy as in Definition \ref{risky_def}, parameterized with threshold value $T$. Also, let the queue-length dynamics be governed by (\ref{brownian_model}). Then, the interruption probability for this policy is given by
\begin{equation}\label{pint_threshold}
    p^T(D) = \left\{
             \begin{array}{ll}
               e^{-\t_1 D} +  p(T) (1- \frac{\t_0}{\t_1})   \big(1- e^{-\t_1 D} \big), & \hbox{$0 \leq D \leq T$} \\
                p(T) e^{-\t_0(D-T)}, & \hbox{$D \geq T$,}
             \end{array}
           \right.
\end{equation}
where
$$p(T) = \frac{e^{-\t_1 T}}{\frac{\t_0}{\t_1} + (1- \frac{\t_0}{\t_1}) e^{-\t_1 T}},$$
and $\theta_i \triangleq 2(R_i - 1)$, for $i \in \{0,1\}$.
\end{theorem}
\begin{proof}
See Appendix \ref{app:fluid_policy_analysis}.
\end{proof}

\begin{corollary}\label{cor:threshold_policy}
Let $\pi^T$ be the threshold policy as in Definition \ref{risky_def}. Then, the policy $\pi^T$ is $(D,\e)$-feasible (cf. Definition \ref{feas_policy_def}) for the following choices of the threshold $T$:
\begin{enumerate}
  \item For all $\e \geq e^{-\t_0 D}$, let $T = 0$.
  \item For all $e^{-\t_1 D} \leq \e \leq e^{-\t_0 D}$, let $T = T(D,\e)$ be the unique solution of $p^T(D) = \e,$ where $p^T(D)$ is given by (\ref{pint_threshold}).
  \item For all other $\e$, there exists no such $T$.
\end{enumerate}

\end{corollary}
\begin{proof}
The proof directly follows from the characterization of the interruption probability in Theorem \ref{thm:threshold_policy}, noting the fact that $p^T(D) \in [e^{-\t_1 D}, e^{-\t_0 D}]$ is monotonically decreasing in $T$.
\end{proof}

Next, we provide a exact characterization of the expected cost of the threshold policy $\pi^T$ for a given threshold $T$. This allows us to obtain a  proper candidate solution for the HJB equation.

\begin{theorem}\label{thm:threshold_cost}
Let $\pi^T$ be the threshold policy as in Definition \ref{risky_def}. Define $J^T(D)$ as the expected cost associated with policy $\pi^T$ given the initial condition $D$ for queue-length dynamics (\ref{brownian_model}) and threshold $T$. The \emph{cost-to-go} function $J^T(D)$ is given by
\begin{equation}\label{J_D}
    J^T(D) = \left\{
               \begin{array}{ll}
                 e^{-\t_0(D - T)} J(T), & \hbox{$D \geq T$} \\
                 (J(T) + \frac{2}{\t_1}T)\frac{1-e^{-\t_1 D}}{1-e^{-\t_1 T}} - \frac{2}{\t_1}D, & \hbox{$D \leq T$,}
               \end{array}
             \right.
\end{equation}
where
\begin{equation}\label{J_T}
    J(T) = \frac{\frac{2}{\t_1^2}\big[1 - (1+\t_1 T)e^{-\t_1 T}\big]}{\frac{\t_0}{\t_1} + (1- \frac{\t_0}{\t_1}) e^{-\t_1 T}}.
\end{equation}
\end{theorem}
\begin{proof}
See Appendix \ref{app:fluid_policy_analysis}.
\end{proof}

The following theorem provide a candidate for value function and verify the optimality condition given by HJB equation in (\ref{HJB_fluid2}).

\begin{theorem}\label{thm:value_function_HJB}
For all $(Q,p) \in \mathcal R$, define
\begin{equation}\label{value_function_candidate}
    V(Q,p) = J^{T(Q,p)}(Q),
\end{equation}
where $J^T(\cdot)$ is defined in (\ref{J_D}), $\mathcal R$ is defined in (\ref{Region_def}), and $T(Q,p)$ is the unique solution of
$$p^T(Q) = p.$$ Then, the HJB equation (\ref{HJB_fluid2}) and boundary condition (\ref{BC_valuefunction}) hold for all $(Q,p) \in \mathcal R$.
\end{theorem}
\begin{proof}
See Appendix \ref{app:fluid_policy_analysis}.
\end{proof}

Theorem \ref{thm:value_function_HJB} verifies that the value function $V(Q,p)$ given by (\ref{value_function_candidate}) is indeed the optimal cost function defined in (\ref{value_def}). Furthermore, we can conclude that the policy $\pi^*(Q,p)$ achieving the minimum in the HJB equation (\ref{HJB_fluid2}) is optimal. In general, the optimal policy depends on both state variables $(Q,p)$ and is not Markov with respect to $Q$. In the following, we show that the state trajectory steered by the optimal policy is limited to  a one-dimensional manifold and the threshold policy $\pi^T$ is optimal for all $(Q,p) \in \mathcal R$. Recall the policy $\pi^T$ boils down to the optimal policy $\pi_0 \equiv 0$ for all other admissible states, by using threshold value $T=0$.

\begin{theorem}\label{thm:threshold_optimalit}
Let $\pi^*(Q,p)$ attain the minimum in the HJB equation (\ref{HJB_fluid2}) for any $(Q,p) \in \mathcal R$. Let $(Q^*_t, p^*_t)$ denote the state trajectory given the initial condition $(D,\e)$, under the control trajectory $(u^*_t, \ph^*_t) = \pi^*(Q^*_t, p^*_t)$. Then, the state trajectory is limited to a one-dimensional invariant manifold $\mathcal M(D,\e)$, where
\begin{eqnarray}\label{manifold_fluid}
&& \mathcal M{(D,\e)}= \big\{(Q,p): p = p^{T(D,\e)}(Q) \big\},
\end{eqnarray}
where $T(D, \e)$ is the solution of $p^T(D) = \e$, and $p^T(\cdot)$ is defined in (\ref{pint_threshold}). Moreover, the optimal policy $\pi^*(Q,p)$ coincides with the threshold policy $\pi^{T(D,\e)}(Q)$.
\end{theorem}
\begin{proof}
See Appendix \ref{app:fluid_policy_analysis}.
\end{proof}
%

Figure \ref{manifold_fig} illustrates a conceptual figure describing the intuition behind Theorem \ref{thm:threshold_optimalit}. Observe that the optimal policy satisfying the HJB equation, divides the feasible state space into two sub-regions corresponding to $u^* =0$ and $u^* = 1$, i.e., the policy switches costly server on/off when the state of the system crosses the boundary between these sub-regions. Theorem \ref{thm:threshold_optimalit} states that for any initial condition $(Q_0, p_0) = (D,\e)$ the state trajectory lies on a one-dimensional manifold $\mathcal M(D,\e)$. Figure \ref{manifold_fig} illustrates these manifolds for different initial conditions. Since the state trajectory is limited to  a one-dimensional space, the decision of switching to the costly server merely depends on the queue-length process. The proof of Theorem \ref{thm:threshold_optimalit} in Appendix  \ref{app:fluid_policy_analysis}, shows that the queue-length at the switch point for each manifold $\mathcal M(D,\e)$ coincides with the threshold $T(D,\e)$ specified in Corollary \ref{cor:threshold_policy}.

\begin{figure}[htbp]
\centering
  \includegraphics[width=4in]{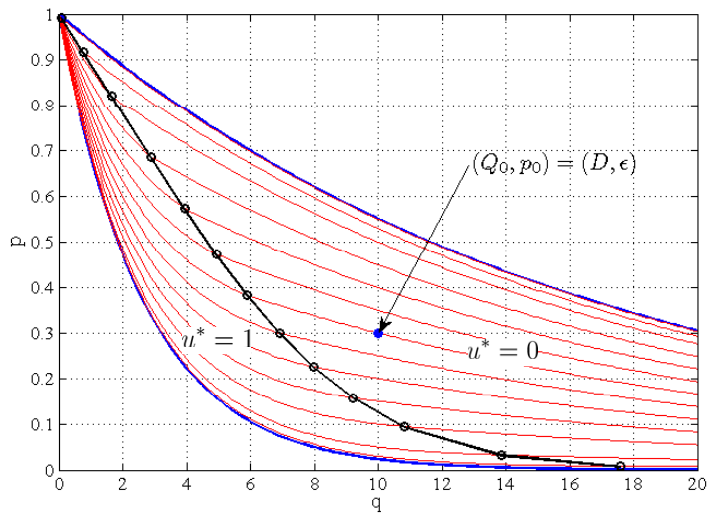}
  \caption{Trajectory of the optimal policy lies on a one-dimensional manifold.}\label{manifold_fig}
\end{figure}

\section{Conclusions and Future Work}\label{sec:conclusion}

We presented a new framework for studying media streaming systems in volatile environments, with focus on quality of user experience. We proposed two intuitive metrics that essentially capture the notion of \emph{delay} from the end-user's point of view. The proposed metrics in the context of media streaming, are initial buffering delay, and probability of interruption in media playback.   These metrics are tractable enough to be used as a benchmark for system design.

We first addressed the problem of streaming in a  technology-heterogeneous multi-server system. The main challenge in multi-server systems is inefficiencies in multi-path streaming due to duplicate packet reception. This issue can also significantly complicates the analysis. We proposed random linear network coding as the solution to this challenge. By sending random linear combination of packets, we remove the notion of identity from packets and hence, guarantee that no packet is redundant. Using this approach allows us to significantly simplify the flow control of multi-path streaming scenarios, and model heterogeneous multi-server systems as a single-server system.

Equipped with tools provided by network coding, we added another level of complexity to the multi-server system. We used our framework to study multi-server systems when the access cost varies across different servers. Our objective was to investigate the trade-offs between the network usage cost and the user's QoE requirements parameterized by initial waiting time and allowable probability of interruption in media playback.  For a Poisson arrival model, we analytically characterized and compared the expected cost of both off-line and online policies, finally showing that a threshold-based online risky policy achieves the lowest cost. The threshold policy uses the costly server if and only if the receiver's buffer is below a certain threshold. Moreover, we observed that even rare but properly timed usage of alternative access technologies significantly improves user experience without any bandwidth over-provisioning.

 We formulated the access cost minimization problem as a Markov decision problem with probabilistic constraints, and characterized the optimal policy by the HJB equation. For a fluid approximation model, we established the optimality of a threshold-based online policy in the class of deterministic Markov policies using the HJB equation as a verification method.

The framework that we have developed in this work can also be used to design adaptive resolution streaming systems that not only depend on the channel conditions but also on the delay requirement of the application, which is captured by the queue-length at the receiver. As for other extensions of this work, we would like to study more accurate models of channel variations such as the two-state Markov model due to Gillbert and Elliot. In this work we focused on deterministic network association policies. Another extension of this work would consist of studying randomized control policies.

\appendices
\section{Analysis of the Control Policies for the Poisson Arrival Model}\label{app:policy_analysis}

\begin{bproof} \textbf{of Theorem \ref{offline_range_thm}. }
By Definition \ref{feas_policy_def}, we need to show that $p^\pi(D) \leq \e$. By a union bound on the interruption probability, it is sufficient to verify
\begin{equation}\label{union_bd}
\pr\Big(\min_{0\leq t\leq t_s} Q_t \leq 0 \big| Q_0 = D\Big) + \pr\Big(\min_{t > t_s} Q_t \leq 0 \big| Q_0 = D\Big) \leq \e.
\end{equation}

In the interval $[0,t_s]$,  $Q_t$ behaves as in a single-server system with rate $R_1$. Hence, by Lemma \ref{pd_exact_lemma} we get
\begin{equation}\label{pint_bd1}
\pr\Big(\min_{0\leq t\leq t_s} Q_t \leq 0 \big| Q_0 = D\Big) \leq e^{-\alpha_1 D}.
\end{equation}

For the second term in (\ref{union_bd}), we have

\begin{eqnarray*}
  \pr\Big(\min_{t > t_s} Q_t \leq 0 \big| Q_0 = D\Big)   &=& \sum_{q = D-t_s}^\infty \pr\Big(\min_{t > t_s} Q_t \leq 0 \big| Q_{t_s} = q\Big) \pr(Q_{t_s} = q) \\
&  \stackrel{(a)}{\leq}& \sum_{q = D-t_s}^\infty e^{-\a_0 q} \pr(Q_{t_s} = q)   \\
&=& \sum_{k = 0}^\infty e^{-\a_0 (D+k-{t_s})} \pr(N_{t_s} + N^c_{t_s} = k) \\
&  \stackrel{(b)}{=}& \sum_{k = 0}^\infty e^{-\a_0 (D+k-{t_s})} \frac{e^{-R_1 {t_s}}(R_1 {t_s})^k}{k!}\\
 &=& e^{-\a_0 (D-{t_s})+R_1 {t_s} (e^{-\a_0}-1)} \sum_{k = 0}^\infty \frac{e^{-R_1 {t_s} e^{-\a_0}}(R_1 {t_s} e^{-\a_0})^k}{k!}\\
&=& \exp\Big(-\a_0 (D-{t_s})+R_1 {t_s} (e^{-\a_0}-1)\Big)\cdot 1 \\
 &\stackrel{(c)}{=}& \exp\Big(-\a_0 (D-{t_s})+R_1 {t_s} (-\frac{\a_0}{R_0})\Big) \\
 &\stackrel{(d)}{\leq}& \e - e^{-\alpha_1 D},
\end{eqnarray*}
where (a) follows from Lemma \ref{pd_exact_lemma} and the fact that $u_t = 0$, for $t \geq t_s$. (b) is true because $N_{t_s} + N^c_{t_s}$ is a Poisson random variable with mean $R_1 t_s$. (c) holds since $\a_0 = I(R_0)$ is the root of $\gamma(r) =  r + R_0(e^{-r} -1 )$. Finally, (d) follows from the hypothesis of the theorem.

By combining the above bounds, we may verify (\ref{union_bd}) which in turns proves feasibility of the proposed control policy.
\end{bproof}

\vspace{0.2in}
\begin{bproof} \textbf{of Theorem \ref{safe_thm}. }
Similarly to the proof of Theorem \ref{offline_range_thm}, we need to show that the total probability of interruption before and after crossing the threshold $S$ is bounded from above by $\e$. Observe that for any realization of $\tau_S$ the bound in (\ref{pint_bd1}) still holds. Further, since the costly server is not used after crossing the threshold and $Q_{\tau_S} \geq S$, Lemma \ref{pd_exact_lemma} implies
\begin{eqnarray}\label{pint_bd2}
\pr\Big(\min_{t > \tau_S} Q_t \leq 0 \big| Q_0 = D\Big)   \leq  e^{-\alpha_0 S} \leq \e -  e^{-\alpha_1 D},&&
\end{eqnarray}
where the second inequality follows from (\ref{S_range}). Finally, combining (\ref{pint_bd1}) and (\ref{pint_bd2}) gives $p^{\pi^S}(D) \leq \e$, which is the desired feasibility result.

For the second part, first observe that $J^{\pi^{S}}(D,\e) = \E[\tau_S]$. In order to cross a threshold $S \geq S^*$, the threshold $S^*$ must be crossed earlier, because $Q_0 = D \leq S^*$. Hence, $\tau_S$ stochastically dominates $\tau_S^*$, implying
$$J^{\pi^{S}}(D,\e) = \E[\tau_S] \geq \E[\tau_{S^*}] =J^{\pi^{S^*}}(D,\e), \quad \foral S\geq S^*.$$

It only remains to compute $\E[\tau_{S^*}]$. It follows from Wald's identity or Doob's optional stopping theorem \cite{SP_book} that
\begin{equation}\label{Walds_stopping}
D + (R_1 - 1)\E[\tau_{S^*}] = \E[Q_{\tau_{S^*}}] = S^* + \xi,
\end{equation}

where $\xi \in [0,1)$ because the jumps of a Poisson process are of units size, and hence the overshoot size when crossing a threshold is bounded by one, i.e.,  $S^* \leq Q_{\tau_{S^*}} < S^*+1$. Rearranging the terms in (\ref{Walds_stopping}) and plugging the value of $S^*$ from (\ref{S_range}) immediately gives the result.
\end{bproof}


\begin{lemma}\label{stop_lemma}
Let $Q_t$ be the buffer size of a single-server system with arrival rate $R > 1$. Let the initial buffer size be $D$ and for any $T\geq D > 0$ define the following stopping times
\begin{equation}\label{stop_def}
\tau_T = \inf \{t > 0: Q_t \geq T\},\quad  \tau_e = \inf \{t \geq 0: Q_t \leq 0\}.
\end{equation}
Then
\begin{equation}\label{stop_prob}
\pr(\tau_e > \tau_T) = \frac{1-e^{-I(R) D}}{1-\E[e^{-I(R) Q_{\tau_T}}|\tau_e > \tau_T]},
\end{equation}
where $I(R)$ is defined in Lemma \ref{pd_exact_lemma}.
\end{lemma}

\begin{proof}
Let $Y(t) = e^{-I(R) Q_t}$. We may verify that $Y(t)$ is a Martingale and uniformly integrable. Also, define the stopping time $\tau = \min\{\tau_T, \tau_e\}$. Since $R>1$, we have
$\pr(\tau \geq t) \leq \pr(0 < Q_t < T) \rightarrow 0$, as $t\rightarrow \infty$. Hence, $\tau < \infty$ almost surely. Therefore, we can employ Doob's optional stopping theorem \cite{SP_book} to write
\begin{eqnarray*}
e^{-I(R) D} &=& \E[Y(0)] = \E[Y(\tau)] \\
&=& \pr(\tau_e \leq \tau_T) \cdot 1\\
&& + \pr(\tau_e > \tau_T) \E[e^{-I(R) Q_{\tau_T}} |\tau_e > \tau_T].
\end{eqnarray*}
The claim immediately follows from the above relation after rearranging the terms.

\end{proof}


\begin{bproof} \textbf{of Theorem \ref{risky_thm}. }
Let us first characterize the interruption probability of the policy $\pi^T$ when the initial buffer size is $D = T$. In this case,  by definition of $\pi^T$ the behavior of $Q_t$ is initially the same as a single-server system with rate $R_1$ until the threshold $T$ is crossed. Hence,
\begin{eqnarray}
 p^{\pi^{T}}(T) &=& \pr\Big(\min_{t \geq 0} Q_t \leq 0 \big| Q_0 = T\Big)  \nonumber \\
&=& \pr(\tau_e < \tau_{T})\cdot 1 \nonumber \\
&& + \pr(\tau_{T} < \tau_e)\pr\Big(\min_{t \geq \tau_{T}} Q_t \leq 0 \big| \tau_{T} < \tau_e, Q_0 = T\Big) \nonumber\\
&=& \frac{e^{-\a_1 T} - \E[e^{-\a_1 Q_{\tau_{T}}}|\tau_e > \tau_{T}]}{1-\E[e^{-\a_1 Q_{\tau_{T}}}|\tau_e > \tau_{T}]} \nonumber \\
&&+ \frac{\big(1-e^{-\a_1 T}\big)\pr\Big(\min_{t \geq \tau_{T}} Q_t \leq 0 \big| \tau_{T} < \tau_e, Q_0 = T\Big) }{1-\E[e^{-\a_1 Q_{\tau_{T}}}|\tau_e > \tau_{T}]},
\label{eqn:step1}
\end{eqnarray}
where the last equality follows directly from Lemma \ref{stop_lemma}. Further, we have
\begin{eqnarray}\label{pint_cond}
 && \!\!\!\!\!\! \pr\Big(\min_{t \geq \tau_{T}} Q_t \leq 0 \big| \tau_{T} < \tau_e, Q_0 = T\Big) \nonumber = \int_{T}^{T + 1} \pr\Big(\min_{t \geq \tau_{T}} Q_t \leq 0 \big| Q_{\tau_{T}}\Big) d\mu (Q_{\tau_{T}}) \nonumber \\
&& \stackrel{(a)}{=} \int_{T}^{T + 1} \pr\Big(\min_{t \geq 0} Q_t \leq 0 \big| Q_0\Big) d\mu (Q_0) \nonumber \\
&& \stackrel{(b)}{=} \int_{T}^{T + 1} \pr\Big(\min_{t \geq 0} Q_t \leq 0 \big| \min_{t \geq 0} Q_t \leq T, Q_0\Big)\nonumber \pr\big(\min_{t \geq 0} Q_t \leq T| Q_0\big) d\mu (Q_0) \nonumber \\
&& \stackrel{(c)}{=} \int_{T}^{T + 1} p^{\pi^{T}}(T) e^{-\a_0(Q_0 - T)}  d\mu (Q_0) \nonumber \\
 && = \E[e^{-\a_0(Q_{\tau_{T}} - T)}| \tau_{T} < \tau_e] p^{\pi^{T}}(T),
\end{eqnarray}
where $\mu$ denotes the conditional distribution of $Q_{\tau_{T}}$ given $\tau_{T} < \tau_e$. Note that $Q_{\tau_{T}} \in [T, T+1]$ because the size of the overshoot is bounded by one. Further, (a) follows from stationarity of the arrival processes and the control policy, (b) holds because a necessary condition for the interruption event is to cross the threshold $T$ when starting from a point $Q_0 \geq T$. Finally (c) follows from Lemma \ref{pd_exact_lemma} and the definition of the risky policy. The relations (\ref{eqn:step1}) and (\ref{pint_cond}) together result in
\begin{equation}\label{pint_risky}
    p^{\pi^{T}}(T) = \frac{   e^{-\a_1 T} \big(1-\E_\mu[e^{-\a_1 (Q_{\tau_{T}} - T)}]\big) }{ 1- \E_\mu[e^{-\a_0(Q_{\tau_{T}} - T)}] + \kappa },
\end{equation}
where $\kappa = \E_\mu[e^{-\a_0 Q_{\tau_{T}} -(\a_1 - \a_0)T}]  - \E_\mu[e^{-\a_1 Q_{\tau_{T}} }] \geq 0$. Therefore, using the fact that
\begin{equation}\label{fact}
 1- x \leq e^{-x} \leq 1 - x +\frac{x^2}{2}, \quad \foral x \geq 0,
\end{equation}
we can provide the following bound
\begin{eqnarray} \label{pint_riksy_bd}
   p^{\pi^{T}}(T) &\leq & \frac{  e^{-\a_1 T} \big(\a_1 \E_\mu[Q_{\tau_{T}} - T]\big)}{\a_0 \E_\mu[Q_{\tau_{T}} - T]\Big(1 - \frac{\a_0}{2} \cdot \frac{\E_\mu[(Q_{\tau_{T}} - T)^2]}{\E_\mu[Q_{\tau_{T}} - T]} \Big)} \nonumber \\
&\leq& \frac{\a_1}{\a_0(1-\frac{\a_0}{2})} e^{-\a_1 T} \nonumber \\
&=& \beta e^{-\a_1 T},
\end{eqnarray}
where the last inequality holds, since $0 \leq Q_{\tau_{\bar D}} - \bar D \leq 1$.

Now we prove feasibility of the risky policy $\pi^{T^*}$ when $D > \bar D$. Observe that by (\ref{threshold}), $D > T^*$, hence the behavior of the buffer size $Q_t$ is the same as the one in a single-server system with rate $R_0$ until the threshold $T^*$ is crossed. Thus
\begin{eqnarray*}
    p^{\pi^{T^*}}(D)  &=&   \pr\Big(\min_{t \geq 0} Q_t \leq 0 \big| Q_0 = D\Big) \\
&=& \pr\Big(\min_{t \geq 0} Q_t \leq 0 \big| \min_{t \geq 0} Q_t \leq T^*, Q_0=D\Big)  \pr\big(\min_{t \geq 0} Q_t \leq T^*| Q_0=D\big) \\
&=&  p^{\pi^{T^*}}(T^*)  e^{-\a_0(D-T^*)} \\
&\leq&  \beta e^{-(\a_1-\a_0) T^* - \a_0 D} = \e,
\end{eqnarray*}
where the inequality follows from (\ref{pint_riksy_bd}), and the last equality holds by (\ref{threshold}).

Next we verify the feasibility of the policy  $\pi^{T^*}$ for  $D \leq \bar D$. In this case, $D \leq T^*$ and by definition of the risky policy the system behaves as a single-server system with arrival rate $R_1$ until the threshold $T^*$ is crossed or the buffer size hits zero (interruption). Hence, we can bound the interruption probability as follows
\begin{eqnarray*}
  p^{\pi^{T^*}}(D)  &=& \pr(\tau_e < \tau_{T^*}) \cdot 1
 + \pr(\tau_{T^*} < \tau_e) \pr\Big(\min_{t \geq \tau_{T^*}} Q_t \leq 0 \big| \tau_{T^*} < \tau_e, Q_0 = D\Big)\\
& \stackrel{(a)}{=}& 1 - \pr(\tau_{T^*} < \tau_e)\Big(1 - \E_\mu[e^{-\a_0(Q_{\tau_{T^*}} - T^*)}] p^{\pi^{T^*}}(T^*) \Big)\\
& \stackrel{(b)}{=}& 1 - \frac{1-e^{-\a_1 D}}{1-\E_\mu[e^{-\a_1 Q_{\tau_{T^*}} }]} \Big(1 - \E_\mu[e^{-\a_0(Q_{\tau_{T^*}} - T^*)}] p^{\pi^{T^*}}(T^*) \Big)\\
& \stackrel{(c)}{\leq}& 1 - \frac{1-e^{-\a_1 D}}{1-\E_\mu[e^{-\a_1 Q_{\tau_{T^*}} }]} \Big(1 - \E_\mu[e^{-\a_0(Q_{\tau_{T^*}} - T^*)}] \beta e^{-\a_1 T^*} \Big)\\
& \stackrel{(d)}{\leq}& \frac{(\beta-1)(1 - e^{-\a_1 D})}{1-\E_\mu[e^{-\a_1 Q_{\tau_{T^*}} }]} + 1 - \beta(1 - e^{-\a_1 D})\\
& \stackrel{(e)}{\leq}& \frac{(\beta-1)(1 - e^{-\a_1 D})}{1-e^{-\a_1 T^*} } + 1 - \beta(1 - e^{-\a_1 D}) \stackrel{(f)}{=} \e,
\end{eqnarray*}
where (a) follows from (\ref{pint_cond}), (b) is a direct consequence of Lemma \ref{stop_lemma}, (c) is a result of (\ref{pint_riksy_bd}), (d) may be verified by noting that $\a_0 = I(R_0)$, $\a_1 = I(R_1)$ and $R_1 \geq R_0$, (e) holds since $\beta \geq 1$ and $Q_{\tau_{T^*}} \geq T^*$. Finally, (f) immediately follows from plugging in the definition of $T^*$ from (\ref{threshold}).

Therefore, the risky policy $\pi^{T^*}$ is feasible by Definition \ref{feas_policy_def}. Observe that the buffer size under any policy $\pi^T$ of the form (\ref{risky_policy}) with $T\geq T^*$ stochastically dominates that of policy $\pi^{T^*}$, because $\pi^T$ switches to the costly server earlier, and stays in that state longer. Hence, $\pi^T$ is feasible for all $T\geq T^*$.
\end{bproof}

\vspace{0.2in}
\begin{bproof} \textbf{of Theorem \ref{risky_cost_thm}. }
Similarly to the proof of Theorem \ref{risky_thm}, we first consider the risky policy $\pi^T$ with the initial buffer size $T$. By definition of $\pi^{T}$, the costly server is used until the threshold $T$ is crossed. Thus the expected cost of this policy is bounded by the expected time until crossing the threshold plus the expected cost given that the threshold is crossed, i.e.,
\begin{equation*}
   J^{\pi^{T}}(T,\e) \leq \frac{\E[Q_{\tau_{T}}] - T}{R_1-1} + E[e^{-\a_0(Q_{\tau_{T}} - T)}]    J^{\pi^{T}}(T,\e),
\end{equation*}
where $\tau_{T}$ is defined in (\ref{stop_def}). The above relation implies
\begin{eqnarray}\label{risky_cost_bd}
     J^{\pi^{T}}(T,\e) &\leq&  \frac{1}{R_1-1} \cdot \frac{\E[Q_{\tau_{T}} - T] }{1- \E[e^{-\a_0(Q_{\tau_{T}} - T)}]} \nonumber \\
&\leq& \frac{1}{R_1-1} \cdot \frac{\E[Q_{\tau_{T}} - T] }{1 - \E\big[1 - \a_0(Q_{\tau_{T}} - T) + \frac{\a_0^2}{2} (Q_{\tau_{T}} - T)^2 \big] } \nonumber \\
&=& \frac{1}{R_1-1} \cdot \frac{1}{\a_0\Big(1 - \frac{\a_0}{2} \cdot \frac{\E_[(Q_{\tau_{T}} - T)^2]}{\E_[Q_{\tau_{T}} - T]}\Big)} \nonumber \\
&\leq& \!\! \frac{1}{\a_0(R_1-1)(1-\frac{\a_0}{2})} = \frac{\beta}{\a_1(R_1-1)},
\end{eqnarray}
where the second inequality follows from the fact in (\ref{fact}), and the last equality holds by definition of $\beta$. Now for any $D\geq \bar D$ we can write
\begin{eqnarray*}
     J^{\pi^{T^*}}(D,\e)   &=&  \pr\Big(\min_{t \geq 0} Q_t \leq T^* \big | Q_0 = D \Big)      J^{\pi^{T^*}}(T^*,\e) \\
&=& e^{-a_0(D- T^*)} J^{\pi^{T^*}}(T^*,\e)
\end{eqnarray*}
where the inequality holds by Lemma \ref{pd_exact_lemma}. Combining this with (\ref{risky_cost_bd}) gives the result in (\ref{risky_cost1}).

If $D\leq \bar D$, the risky policy uses the costly server until the threshold $T^*$ is crossed at $\tau_{T^*}$ or the interruption event ($\tau_e$), whichever happens first. Afterwards, no extra cost is incurred if an interruption has occurred. Otherwise, by (\ref{risky_cost_bd}) an extra cost of at most $\frac{\beta}{\a_1(R_1-1)} $ is incurred, i.e.,
\begin{equation*}
     J^{\pi^{T^*}}(D,\e) \leq \E\big[\min\{\tau_e, \tau_{T^*}\}\big] + \pr(\tau_{T^*} < \tau_e) \frac{\beta}{\a_1(R_1-1)}.
\end{equation*}

By Doob's optional stopping theorem applied to the Martingale $Z_t = Q_t - (R_1-1)t$, we obtain

$$D = \pr(\tau_{T^*} < \tau_e) \E[Q_{\tau_{T^*}}| \tau_{T^*} < \tau_e] -(R_1-1)\E\big[\min\{\tau_e, \tau_{T^*}\}\big], $$
which implies
$$\E\big[\min\{\tau_e, \tau_{T^*}\}\big] \leq \frac{\pr(\tau_{T^*} < \tau_e) (T^*+1) - D}{R_1-1}.  $$

By combining the preceding relations we conclude that
$$  J^{\pi^{T^*}}(D,\e) \leq \frac{\pr(\tau_{T^*} < \tau_e)}{R_1-1}\Big(T^*+1+\frac{\beta}{\a_1}\Big) - \frac{D}{R_1-1},$$
which immediately implies (\ref{risky_cost2}) by employing Lemma \ref{stop_lemma}.

\end{bproof}

\section{Analysis of the Threshold Policy for the Fluid Approximation Model}\label{app:fluid_policy_analysis}

\begin{bproof} \textbf{of Lemma \ref{lemma:degenerate}. } There are multiple approaches to prove the claim. We prove a more general case using Doob's optional stopping theorem that will be useful in the later arguments. We only consider $i = 0$; the other case is the same.

Let $Y_t = e^{-\t_0 Q_t}$. It is straightforward to show that $Y_t$ is a Martingale with respect to $W_t$. Now consider the boundary crossing problem, where we are interested in the probability of hitting zero before a boundary $b > D$. Let $\tau$ denote the hitting time of either boundaries. For any $n > 0$, we may apply Doob's optional stopping theorem \cite{SP_book} to the stopped Martingale $Y_{\tau \wedge n}$ to write
$$\E[Y_{\tau \wedge b}] = \E[e^{-\t_0 Q_{\tau \wedge n}}] = e^{-\t_0 D}, \quad \foral n.$$

Now, we take the limit as  $n \rightarrow \infty$ and exploit the dominant convergence theorem to establish:
\begin{equation}\label{doob_result}
\E[Y_\tau] = \E[e^{-\t_0 Q_\tau}] = \lim_{n \rightarrow \infty}\E[e^{-\t_0 Q_{\tau \wedge n}}] =  e^{-\t_0 D}.
\end{equation}

Finally, using the Borel-Cantelli Lemma, we can show $\tau$ is finite with probability one, which allows us to decompose (\ref{doob_result}) and characterize the boundary crossing probabilities as
\begin{eqnarray*}
\pr(Q_\tau = 0)\cdot 1 + \pr(Q_\tau = b) \cdot e^{-\t_0 b} &=& e^{-\t_0 D}, \\
\pr(Q_\tau = 0) + \pr(Q_\tau = b) &=& 1.
\end{eqnarray*}

Solving the above equations gives
\begin{eqnarray}\label{crossing0_prob}
\pr(Q_\tau = 0) &=& \frac{e^{-\t_0 D} - e^{-\t_0 b}}{1 - e^{-\t_0 b}}, \\
\pr(Q_\tau = b) &=& \frac{1 -e^{-\t_0 D} }{1 - e^{-\t_0 b}}.\label{crossingb_prob}
\end{eqnarray}

Taking the limit as $b \rightarrow \infty$ proves the claim.
\end{bproof}
\vspace{.2in}

\begin{bproof} \textbf{of Theorem \ref{thm:threshold_policy}. }
We first characterize the interruption probability for the cases $D > T$ and $D < T$ given $p(T)$, which is the interruption probability starting from $D=T$.

For any $x \geq 0$, define $\tau_x$ as the first hitting time of boundary $x$, i.e.,
\begin{equation}\label{hitting_time_def}
\tau_x = \inf\{t\geq 0: Q_t = x\}.
\end{equation}

For the case $D > T$, using path-continuity of $Q_t$, strong Markov property and Lemma \ref{lemma:degenerate}, we have
\begin{eqnarray*}
p(D) &=& \pr(\tau_0 < \infty| Q_0 = D) \\
&=&  \pr(\tau_0 < \infty | Q_0 = T) \cdot \pr(\tau_T < \infty | Q_0 = D)\\
&=& e^{-\t_0(D-T)}p(T).
\end{eqnarray*}

For the case $D < T$, we use the boundary crossing probabilities (\ref{crossing0_prob}) and (\ref{crossingb_prob}) that we derived in the proof of Lemma \ref{lemma:degenerate}. Note that for the threshold policy when $D < T$, the drift is set to $R_1 -1 = 2\t_1$. Hence, by total probability theorem and strong Markov property, we obtain
\begin{eqnarray*}
p(D) &=& \pr(\tau_0 < \infty| Q_0 = D) \\
&=& 1\cdot\pr(\tau_0 < \tau_T | Q_0 = D) + \pr(\tau_0 < \infty | Q_0 = T)\cdot \pr(\tau_T < \tau_0 | Q_0 = D) \\
&=& \frac{e^{-\t_1 D} - e^{-\t_1 b}}{1 - e^{-\t_1 b}} + p(T)\frac{1 -e^{-\t_1 D} }{1 - e^{-\t_1 b}}.
\end{eqnarray*}
We may obtain the desired result after simple manipulations of the above relation, once we compute $p(T)$.

In order to characterize $p(T)$, we use an analogue of one-step deviation analysis for Markov chains. Let $Q_0 = T$, and consider a small deviation $Q_h$, where $h$ is a small time-step. Since $Q_t$ is a Brownian motion with drift, $Q_h$ has a normal distribution with variance $h$, and mean of $T+\a h$, where $\a \in [R_0 -1, R_1 -1]$. Therefore, the probability of $Q_h \geq T$ is $(\frac12 + \delta) + o(h)$, where $\delta$ is a small constant of the same order of $h$, and $\frac{o(h)}{h} \rightarrow 0$ as $h\rightarrow 0$. By strong Markov property of the Brownian motion, (\ref{crossing0_prob}) and (\ref{crossingb_prob}), we have
\begin{eqnarray}\label{pT_equation}
p(T) &=& \pr(\tau_0 < \infty | Q_0 = T) \nonumber \\
&=& \pr(\tau_0 < \infty | Q_h \geq T) (\frac12 + \delta) + \pr(\tau_0 < \infty | Q_h < T) (\frac12 - \delta) + o(h) \nonumber\\
&=& \Big[ 0 + p(T) \E_{Q_h}[\pr(\tau_T < \infty)| Q_h \geq T]  \Big] (\frac12 + \delta) \nonumber\\
 && + \Big[1 \cdot \E_{Q_h}[\pr(\tau_0 < \tau_T)| Q_h < T] + p(T) \E_{Q_h}[\pr(\tau_T < \tau_0)| Q_h < T]\Big] (\frac12 - \delta) \nonumber\\
 && + o(h) \nonumber\\
&=& \Big[ 0 + p(T) \E[ e^{-\t_0 Z} | Z \geq 0]  \Big] (\frac12 + \delta)\nonumber\\
 && + \E\Big[ \frac{e^{-\t_1 T} (e^{-\t_1 Z}-1)}{1-e^{-\t_1 T}} \Big| Z < 0\Big]  (\frac12 - \delta) \nonumber\\
 && + p(T) \E\Big[ 1 - \frac{e^{-\t_1 T} (e^{-\t_1 Z}-1)}{1-e^{-\t_1 T}} \Big| Z < 0\Big] (\frac12 - \delta) + o(h),
\end{eqnarray}
where $Z = Q_h - T$ is a Gaussian random variable with mean $\a h$ and variance $h$. In order to obtain $p(T)$, we need to compute $E[e^{-\t_0 Z} | Z \geq 0]$ and $E[e^{-\t_1 Z} | Z < 0]$. We may compute these expressions exactly, but it is simpler to compute upper and lower bounds and then take the limit as $h\rightarrow 0$. By (\ref{fact}), we have
\begin{eqnarray}\label{Z_bound1}
E[e^{-\t_0 Z} | Z \geq 0] &\leq& \E[1 - \t_0 Z + \frac{(\t_0 Z)^2}{2} | Z \geq 0] = 1 - \t_0 \beta \sqrt{h} + o(\sqrt{h}), \nonumber \\
E[e^{-\t_0 Z} | Z \geq 0] &\geq& \E[1 - \t_0 Z | Z \geq 0] = 1 - \t_0 \beta \sqrt{h} + o(\sqrt{h}),
\end{eqnarray}
where $\beta$ is a constant. Similarly, we get
$$ 1 + \t_1 \beta \sqrt{h} + o(\sqrt{h}) \leq E[e^{-\t_1 Z} | Z < 0] \leq 1 + \t_1 \beta \sqrt{h} + o(\sqrt{h}).
$$

Plugging these relations back in (\ref{pT_equation}), dividing by $\beta \sqrt{h}$ and taking the limit as $h$ goes to zero, we obtain the following equation
\begin{equation}\label{Z_bound2}
p(T) \Big[\t_0 +\frac{e^{-\t_1 T}}{1-e^{-\t_1 T}}\t_1 \Big] = \frac{e^{-\t_1 T}}{1-e^{-\t_1 T}}\t_1,
\end{equation}
which gives the desired result for $p(T)$ after rearranging the terms.
\end{bproof}

\vspace{.2in}
\begin{bproof} \textbf{of Theorem \ref{thm:threshold_cost}. } The proof technique for this theorem is analogous to that of Theorem \ref{thm:threshold_policy}. First, we consider the cases $D>T$ and $D<T$ and characterize the expected cost in terms of $J(T)$. Let $\tau_x$ be defined as in (\ref{hitting_time_def}).

For the case $D > T$, note that no cost is incurred until the threshold $T$ is reached. Hence
\begin{eqnarray*}
J(D) &=&  \E \bigg[\int_0^{\tau_e} u_t dt \Big |Q_0 = D\bigg]\\
 &=&  \E \bigg[\int_0^{\tau_e} u_t dt \Big | \tau_T = \infty, Q_0 = D \bigg] \pr(\tau_T = \infty | Q_0 = D) \\
 && + \E \bigg[\int_0^{\tau_e} u_t dt \Big | \tau_T < \infty , Q_0 = D\bigg] \pr(\tau_T < \infty  | Q_0 = D) \\
 &= & 0 +   \E\bigg[0 +\int_{\tau_T}^{\tau_e} u_t dt \Big | \tau_T < \infty, Q_0 = D \bigg] \pr(\tau_T < \infty  | Q_0 = D)\\
 &\stackrel{(a)}{=}& \E\bigg[\int_{0}^{\tau_e} u_t dt \Big | Q_0 = T \bigg] \pr(\tau_T < \infty  | Q_0 = D)\\
 &=&J(T) \pr(\tau_T < \infty| Q_0 = D) \stackrel{(b)}{=} J(T)e^{-\t_0(D-T)},
\end{eqnarray*}
where $(a)$ follows from the memoryless property of Brownian motion and $(b)$ is a consequence of Lemma \ref{lemma:degenerate}.

For the case $D < T$, we can use a strong Markov property to write the following for a small time-step $h$:
\begin{eqnarray*}
J(D) &=& J(Q_0) = 1\cdot h + \E_W[J(Q_h)]  \\
&=& h + \E_W\Big[J(D) + \frac{\partial J}{\partial D} ((R_1 - 1)h + W_h) +  \frac12 \cdot \frac{\partial^2 J}{\partial D^2} \cdot h \Big] + o(h) \\
&=& h + J(D) + (R_1 - 1) \frac{\partial J}{\partial D}\cdot h +  \frac12 \cdot \frac{\partial^2 J}{\partial D^2} \cdot h + o(h),
\end{eqnarray*}
which gives the following ordinary differential equation after dividing by $h$ and taking the limit as $h\rightarrow 0$
\begin{equation}\label{J_ODE}
    \frac{\partial^2 J}{\partial D^2}  + \t_1 \frac{\partial J}{\partial D} + 2 = 0,\quad 0\leq D\leq T.
\end{equation}

It is straightforward to solve the differential equation in (\ref{J_ODE}) with the boundary condition $J(0) = 0$, and $J(T)$ as a parameter. This result completes the characterization of $J(D)$ described in (\ref{J_D}) as a function of $J(T)$. Also, note that if we set the boundary condition $J(T) = 0$, $J(D)$ gives the expected time to hit either of the boundaries at $0$ or $T$, i.e., we get
\begin{equation}\label{E_minT}
    \E\Big[\min\{\tau_0, \tau_T\}\big| Q_0 = D\Big] = \frac{2}{\t_1}\bigg[T \cdot \frac{1- e^{-\t_1 D}}{1- e^{-\t_1 T}} - D \bigg]
\end{equation}

Now, we use a similar technique as in the proof of Theorem \ref{thm:threshold_policy} to compute $J(T)$. Consider a small time step deviation $h > 0$ from the initial condition $Q_0 = T$. Similarly to (\ref{pT_equation}), we have
\begin{eqnarray*}
J(T) &=& \gamma h + \Big[ J(T) \E_{Q_h}[\pr(\tau_T < \infty)| Q_h \geq T]  \Big] \pr(Q_h \geq T | Q_0 = T) \\
 && + \bigg[ 1\cdot \E\Big[\min\{\tau_0, \tau_T\}\big| Q_h < T\Big] + 0 \cdot \E_{Q_h}[\pr(\tau_0 < \tau_T)| Q_h < T] \\
 && \qquad + J(T) \E_{Q_h}[\pr(\tau_T < \tau_0)| Q_h < T]\bigg]  \pr(Q_h < T | Q_0 = T) +o(h)\\
&=& \gamma h + J(T) \E[ e^{-\t_0 Z} | Z \geq 0]  (\frac12 + \delta)\nonumber\\
 && +  \E\Big[\min\{\tau_0, \tau_T\}\big| Q_h < T\Big] (\frac12 - \delta) \nonumber\\
 && + J(T) \E\Big[ 1 - \frac{e^{-\t_1 T} (e^{-\t_1 Z}-1)}{1-e^{-\t_1 T}} \Big| Z < 0\Big] (\frac12 - \delta) + o(h),
\end{eqnarray*}
where $\gamma$ is a constant bounded by 1, $\delta = \Theta(h)$, and $Z = Q_h -T$ is a Gaussian random variable with mean $\a h$ and variance $h$ for some constant $\a$. The second inequality in the preceding relations follows from (\ref{crossingb_prob}) and Lemma \ref{lemma:degenerate}. By (\ref{E_minT}), applying the bounds in (\ref{Z_bound1}) and (\ref{Z_bound2}), dividing by $\beta \sqrt{h}$ and taking the limit at $h\rightarrow 0$, we obtain the following equation
$$J(T)\bigg[\t_1 \frac{e^{-\t_1 T}}{1- e^{-\t_1 T}} + \t_0\bigg] = \frac{2}{\t_2} \bigg[1- \t_1 T \frac{e^{-\t_1 T}}{1- e^{-\t_1 T}} \bigg], $$
which gives us the desired expression in (\ref{J_T}) after rearranging the terms.
\end{bproof}

\vspace{.2in}
\begin{bproof} \textbf{of Theorem \ref{thm:value_function_HJB}. }
In order to facilitate verification of the HJB equation (\ref{HJB_fluid2}), we rewrite and slightly manipulate the candidate solution $V(Q,p)$ given by (\ref{value_function_candidate}). Recall that
\begin{equation}\label{V_candidate}
V(Q,p) =  \left\{
               \begin{array}{ll}
                 V_0(Q,p), & \hbox{$Q \geq T(Q,p)$} \\
                 V_1(Q,p), & \hbox{$Q \leq T(Q,p)$,}
               \end{array}
             \right.
\end{equation}
where
\begin{equation}\label{V0}
   V_0(Q,p) = e^{-\t_0(Q - T(Q,p))} J(T(Q,p)),
\end{equation}
\begin{equation}\label{V1}
   V_1(Q,p) = \big[J(T(Q,p)) + \frac{2}{\t_1}T(Q,p)\big]\frac{1-e^{-\t_1 Q}}{1-e^{-\t_1 T(Q,p)}} - \frac{2}{\t_1}Q,
\end{equation}
\begin{equation}\label{J_TQp}
    J(T(Q,p)) = \frac{\frac{2}{\t_1^2}\big[1 - (1+\t_1 T(Q,p))e^{-\t_1 T(Q,p)}\big]}{\frac{\t_0}{\t_1} + (1- \frac{\t_0}{\t_1}) e^{-\t_1 T(Q,p)}}.
\end{equation}

Note that $p^{T(Q,p)}(Q) = p$; and by definition of $p^T(\cdot)$ in (\ref{pint_threshold}) we may verify that the condition $Q \gtrless T(Q,p)$ is equivalent to
$p \gtrless \frac{e^{-\t_1 Q}}{\frac{\t_0}{\t_1} + (1- \frac{\t_0}{\t_1}) e^{-\t_1 Q}}.$
Therefore, we can partition the feasible region $\mathcal R$ into two sub-regions $\mathcal R_0$ and $\mathcal R_1$, such that
\begin{eqnarray*}
\mathcal R_0 &=& \big\{(Q,p): p \geq \frac{e^{-\t_1 Q}}{\frac{\t_0}{\t_1} + (1- \frac{\t_0}{\t_1}) e^{-\t_1 Q}}\big\} \cap \mathcal R, \\
\mathcal R_1 &=& \big\{(Q,p): p < \frac{e^{-\t_1 Q}}{\frac{\t_0}{\t_1} + (1- \frac{\t_0}{\t_1}) e^{-\t_1 Q}}\big\} \cap \mathcal R.
\end{eqnarray*}
Hence, we need to verify HJB for two regions separately, using the proper expression in (\ref{V_candidate}).

In order to verify the HJB equation for the candidate solution (\ref{value_function_candidate}), we also need to characterize the optimal value of the minimization problem in (\ref{HJB_fluid2}). First, we characterize the optimal solution pair $(u^*, \ph^*)$ for any feasible state $(Q,p)\in \mathcal R$.
 Observe that the optimization problem  in (\ref{HJB_fluid2}) can be decomposed into two smaller problems:
\begin{eqnarray} \label{u_opt_problem}
  u^*(Q,p) &=& \textrm{argmin}_{u \in \{0,1\}} \bigg\{ u +  \frac{\partial V}{\partial Q} (R_{u} - 1) \bigg\}, \\
  \ph^*(Q,p) &=& \textrm{argmin}_{\ph} \bigg\{ \frac12 \frac{\partial^2 V}{\partial p^2} (\ph)^2 + \frac{\partial^2 V}{\partial Q \partial p} \ph \bigg\} \label{p_opt_problem}.
\end{eqnarray}

The minimization problem in (\ref{p_opt_problem}) is quadratic and hence convex in $\ph$. So we can use first order optimality condition to get
\begin{equation}\label{p_opt}
    \ph^*(Q,p) = - {\frac{\partial^2 V}{\partial Q \partial p}(Q,p)} / {\frac{\partial^2 V}{\partial p^2}(Q,p) }.
\end{equation}

For the problem in (\ref{u_opt_problem}), $u^*(Q,p) = 0$ is and only if
$$0 + \frac{\partial V}{\partial Q} (R_{0} - 1)  \leq 1 + \frac{\partial V}{\partial Q} (R_{1} - 1), $$
or equivalently
\begin{equation}\label{u_opt_cond}
\frac{\partial V}{\partial Q}(Q,p) \geq - \frac{1}{R_1 - R_0}.
\end{equation}

Using the chain rule and the implicit function theorem, we can analytically calculate $\frac{\partial V}{\partial Q}$ from (\ref{V_candidate}) to conclude that the condition in (\ref{u_opt_cond}) holds if and only if $(Q,p) \in \mathcal R_0$. In other words, $u^*(Q,p) = 0$ for all $(Q,p) \in \mathcal R_0$ and $u^*(Q,p) = 1$ for all $(Q,p) \in \mathcal R_1$.

In summary, the HJB equation in (\ref{HJB_fluid2}) boils down to the following equations:
\begin{eqnarray}\label{HJB_partition1}
0 &=&   \frac{\t_0}{2} \frac{\partial V_0}{\partial Q} + \frac12 \frac{\partial^2 V_0}{\partial Q^2} - \frac12 \Big(\frac{\partial^2 V_0}{\partial Q \partial p}\Big)^2 / \frac{\partial^2 V_0}{\partial p^2}  , \quad \foral (Q,p) \in \mathcal R_0, \\
0 &=&  1 + \frac{\t_1}{2} \frac{\partial V_1}{\partial Q} + \frac12 \frac{\partial^2 V_1}{\partial Q^2} - \frac12 \Big(\frac{\partial^2 V_1}{\partial Q \partial p}\Big)^2 / \frac{\partial^2 V_1}{\partial p^2}  , \quad \foral (Q,p) \in \mathcal R_1,  \label{HJB_partition2}
\end{eqnarray}
where $V_0(Q,p)$ and $V_1(Q,p)$ are given by (\ref{V0}) and (\ref{V1}), respectively. The verification of (\ref{HJB_partition1}) and (\ref{HJB_partition2}) is straightforward but tedious. We omit the details for brevity. We may simply use symbolic analysis tools such as Mathematica for this part.
\end{bproof}

\vspace{.2in}
\begin{bproof} \textbf{of Theorem \ref{thm:threshold_optimalit}. }
From the proof of Theorem \ref{thm:value_function_HJB}, we have characterized the optimal policy $\pi^*:\R \times [0,1] \rightarrow \{0,1\}\times \R$ as
$$\pi^*(Q,p) = \Big(u^*(Q,p), \ph^*(Q,p)\Big),$$
where $u^*(Q,p) = 0$ if and only if $(Q,p) \in \mathcal R_0$ and $\ph^*(Q,p)$ given by (\ref{p_opt}) can be  explicitly computed as follows:
\begin{eqnarray}\label{ph_opt_explicit}
\ph^*(Q,p) &=& - \t_0 p, \quad \foral (Q,p) \in \mathcal R_0,\nonumber \\ \vspace{.1in}
\ph^*(Q,p) &=& - \bigg[\frac{(1-p)e^{\t_1 Q}}{(1-p e^{\t_1 Q})^2} -\frac{1}{1-p e^{\t_1 Q}} \bigg]/ \bigg[\frac{(1-p)(e^{2\t_1 Q} - e^{\t_1 Q})}{\t_1(1-p e^{\t_1 Q})^2} -\frac{2(e^{\t_1 Q} - 1)}{\t_1(1-p e^{\t_1 Q})} \bigg]\nonumber \\
&=& -\frac{\t_1}{(1+p)e^{\t_1 Q} -2},\hspace{140pt} \foral (Q,p) \in \mathcal R_1.
\end{eqnarray}

Moreover, the dynamics of the state process under the optimal control policy is given by:
\begin{eqnarray*}
dQ^*_t &=& (R_{u^*(Q^*_t, p^*_t)} - 1) dt + dW_t, \\
dp^*_t &=& \ph^*(Q^*_t, p^*_t) dW_t.
\end{eqnarray*}

For the proof of the first claim, observe that for a given manifold  $\mathcal M(D,\e)$, we have
\begin{equation}\label{T_constant_claim}
    T(Q,p) = T(D,\e), \quad \foral (Q,p) \in \mathcal M(D,\e).
\end{equation}
This claim holds by definition of $\mathcal M(D,\e)$ in (\ref{manifold_fluid}), and the fact that $T(Q,p)$ is the unique solution of $p^{T}(Q) = p$. Next, we show that if $(Q^*_t, p^*_t) \in \mathcal M(D,\e)$ for any $t\geq 0$, then after executing the optimal policy $\pi^*$,  the state process stays on the manifold $\mathcal M(D,\e)$. First, consider the case where $Q^*_t \geq T(Q^*_t,p^*_t) = T(D,\e)$. In this case, $u^*(Q^*_t,p^*_t) = 0$ and $ \ph^*(Q^*_t, p^*_t) = -\t_0 p^*_t$. We would like to show that the solution of the stochastic differential equation $dp^*_t = -\t_0 p^*_t dW_t$ coincides with the invariant manifold given by $\tilde p_t = e^{-\t_0(Q^*_t - T(D,\e))}p(T(D,\e))$. By employing It\={o}'s Lemma we can check that $e^{-\t_0(Q^*_t - T(D,\e))}p(T(D,\e))$ is indeed the desired solution. In particular, using the proper evolution of the queue-length process $Q^*_t$, we can write
\begin{eqnarray*}
   d\tilde p_t  &=& -\t_0 e^{-\t_0(Q^*_t - T(D,\e))}p(T(D,\e)) (dQ^*_t) \\
   && + \frac12 \t_0^2 e^{-\t_0(Q^*_t - T(D,\e))}p(T(D,\e)) dt \\
   &=& e^{-\t_0(Q^*_t - T(D,\e))}p(T(D,\e)) \Big[ -\t_0 (\frac{\t_0}{2} dt + dW_t) +\frac12 \t_0^2 \Big] \\
   &=& -\t_0 e^{-\t_0(Q^*_t - T(D,\e))}p(T(D,\e)) dW_t \\
   &=& -p^*_t dW_t.
\end{eqnarray*}

Next, consider the case where  $Q^*_t < T(Q^*_t,p^*_t) = T(D,\e)$. In this case, $u^*(Q^*_t,p^*_t) = 1$ and $ \ph^*(Q^*_t, p^*_t)$ is given by (\ref{ph_opt_explicit}). Similarly to the previous case, we may use  It\={o}'s Lemma to verify that the state process stays on the invariant manifold given by
$$ \tilde p_t = p^{T(D,\e)}(Q^*_t) = e^{-\t_1 Q^*_t} + p(T(D,\e))(1 - \frac{\t_0}{\t_1})\Big(1 -e^{-\t_1 Q^*_t}\Big).$$

By  It\={o}'s Lemma we have
\begin{eqnarray*}
   d\tilde p_t &=& \frac{\partial p^{T(D,\e)}(Q^*_t)}{\partial Q} dQ^*_t + \frac12 \cdot \frac{\partial^2 p^{T(D,\e)}(Q^*_t)}{\partial Q^2} dt \\
   &=& -\t_1 e^{-\t_1 Q^*_t} \Big[1 - p(T(D,\e))(1 - \frac{\t_0}{\t_1})\Big] (\frac{\t_1}{2} dt + dW_t) \\
   && + \frac12 \cdot \t_1^2 e^{-\t_1 Q^*_t} \Big[1 - p(T(D,\e))(1 - \frac{\t_0}{\t_1})\Big] dt \\
    &=&    -p^*_t dW_t = dp^*_t,
\end{eqnarray*}
which completes the proof of the first claim.

Now that we have established that the state process starting from $(D,\e)$ under optimal control stays on a one-dimensional invariant manifold $\mathcal M(D,\e)$, the optimality of the threshold policy $\pi^{T(D,\e)}(Q)$ is immediate. Recall that the decision process of importance is $u^*_t\in \{0,1\}$, and we know that $u^*(Q,p) = 0$ if and only if $Q \geq T(Q,p)$. Moreover, since the optimal state process stays on $\mathcal M(D,\e)$, we have $T(Q^*_t, p^*_t) = T(D,\e)$. Hence, the optimal control policy (given the initial condition) chooses the action $u^*(Q,p) = 0$ if and only if $Q \geq T(D,\e)$. Therefore, the optimal policy $\pi^*(Q,p)$ coincides with the threshold policy $\pi^{T(D,\e)}(Q)$. We may also verify that the interruption probability under the threshold policy conditioned on the history up to time $t$ is given by $p^*_t$.
\end{bproof}

\bibliographystyle{unsrt}
\bibliography{TAC}

\end{document}